\documentclass[11pt]{article}
\usepackage[margin=1in]{geometry}
\usepackage{amsmath,amssymb,amsthm}
\usepackage{mathtools}
\usepackage{graphicx}
\usepackage{enumitem}
\usepackage{microtype}
\usepackage{float}
\usepackage{tikz}
\usepackage{pgfplots}
\pgfplotsset{compat=1.17}

\title{Lund–Regge Geometry and Integrability of a Generalized Konno–Oono System}
\author{Jos\'e Luis D\'iaz Palencia \\
Department of Mathematics and Education\\ Universidad a Distancia de Madrid (UDIMA)\\ Madrid, Spain.\\
joseluis.diaz.p@udima.es \\
and\\
Enrique G. Reyes\\
Departamento de Matem\'atica y Ciencia de la Computaci\'on\\
Universidad de Santiago de Chile\\
Santiago, Chile. \\
enrique.reyes@usach.cl; e\_g\_reyes@yahoo.ca}

\newtheorem{theorem}{Theorem}
\newtheorem{definition}{Definition}
\newtheorem{remark}{Remark}

\newtheorem{proposition}{Proposition}
\newtheorem{corollary}{Corollary}

\begin{document}
\maketitle

\begin{abstract}
We extend recent work on the relation between classical surface theory and partial differential equations, focusing on equations of
pseudo-spherical type in the sense of Chern--Tenenblat and on a non-trivial generalization motivated by the Lund--Regge system describing
surfaces immersed in $S^3$. As our main application, we study a generalized Konno--Oono system with three dependent
variables introduced in a previous paper by one of the authors. We construct an associated parameter-dependent overdetermined 
linear problem and {\em we establish the existence of infinitely many non-trivial local conservation laws}, hence, integrability. The latter result is the most technically demanding result of this paper: 
it requires a refined analysis of a Riccati pseudo-potential expansion, the use of ``stereographic coordinates'' at the full equation manifold level, the construction of special 
representatives, and an explicit cohomological argument  that excludes horizontal exactness. We also analyse an illustrative class of travelling wave solutions and show 
that they can be used to generate surfaces immersed in $S^3$ whose Gaussian curvature changes sign periodically, while their mean curvature are non-vanishing periodic functions.
In a limit case, we obtain surfaces that are locally congruent to generalized Clifford tori.   
\end{abstract}

\section{Introduction}

A classical source of interaction between geometry and integrable partial differential equations is the 
construction of interesting nonlinear evolution equations from surface theory ({\em e.g.}, the sine-Gordon equation)
 as well as the ``inverse problem" of realizing that a given equation of importance for Mathematical Physics ({\em 
 e.g.}, the Korteweg-de Vries or Camassa-Holm equations) can be understood as a necessary and 
 sufficient condition of geometric origin. We begin this introduction by clarifying these two directions of research.

Let us consider three papers from the early days of modern research on interactions between 
differential geometry and integrability: the paper \cite{L} on the Lund-Regge equation, Sasaki's paper \cite{S}, 
and G\"urses and Nutku's paper \cite{GN}:

\begin{enumerate}
\item In \cite{L}, the author shows that the Lund-Regge equation (which we present explicitly later on) can 
be understood as a Gauss-Codazzi equation determining the (local) immersion of a surface into the three-sphere
$S^3$. As a consequence, the Lund-Regge equation is the integrability condition of an $so(3)$-valued 
overdetermined linear problem. This is an example of a ``necessary and sufficient condition" of geometric origin. 
\item Similarly, in \cite{GN} the authors recall a classical differential geometric construction by Guichard that
produces generalizations of pseudo-spherical surfaces, and they construct an $sl(2)$-valued connection 
one-form that is flat on solutions to an interesting system that contains the Korteweg-de Vries equation as a 
special case. Some further comments on \cite{GN} appear in \cite{Reyes2000}. 
\item Finally, in \cite{S}, the author shows that $sl(2)$-valued connection one-forms which are flat on
solutions to an equation solvable by AKNS scattering/inverse scattering, can be used to construct Riemannian metrics of
Gaussian curvature $K=-1$. This is another example of a ``necessary and sufficient condition" of geometric origin.
\end{enumerate}

In this paper we consider two geometric structures that in a sense reverse the perspective of the examples above. Roughly (formal definitions appear in Section 2), we 
say that an equation (or system of equations) is of pseudo-spherical type if it is the necessary and sufficient 
condition for the existence
of a Riemannian metric of Gaussian curvature $K=-1$ (an {\em intrinsic} condition), and that it is of Lund-Regge 
type if it is the necessary and sufficient condition for the existence of surfaces locally immersed in $S^3$ (an
{\em extrinsic} condition). For the first structure, a short summary of its main characteristics and 
applications will suffice: it is already well-known, as it was introduced by S.S. 
Chern and K. Tenenblat in 1986, see \cite{CT}. For the second structure we go into more detail since it has 
been introduced only recently, see \cite{BHR}, and much remains to be done: in addition to presenting the main features 
of this new structure and showing how it is motivated by the Chern and Tenenblat theory, we study the 
integrability and travelling wave dynamics of one of the
examples appearing in \cite{BHR}, a generalized form of the Konno--Oono equation given by 
\begin{align*}
  q_t + 2 r r_x &= 0  \\
  r_{xt} - 2qr - 4r\delta_x\delta_t &= 0 \\ \tag{KO}
  \delta_{xt} r + r_t\delta_x + r_x\delta_t &= 0\; . 
\end{align*}

The system (KO) is a natural object of study since it is, to the best of our knowledge, a new example of an integrable system of partial differential equations with an interesting 
geometric interpretation. This system is not presented merely as another equation 
admitting a zero-curvature representation: one of our main technical contributions in this paper is the proof that infinitely many conservation laws 
produced by a Riccati pseudopotential expansion associated to (KO) are genuinely non-trivial (and linearly independent) in horizontal cohomology. This proof explicitly uses the
cohomological interpretation of conservations laws (see for instance \cite[Section 4.3]{O}). It requires a detailed 
analysis of a quadratic pseudopotential formally deduced from its zero curvature representation, an understanding of the geometry of the full equation manifold of (KO), 
the selection of a suitable subsequence of horizontally closed one-forms, the construction of appropriate representatives, and an explicit cohomological non-triviality argument.
Thus, our arguments are essentially different to other checks of non-triviality appearing in the literature, see for instance \cite{R-ch}.
Our constructions also show that the original Konno--Oono system describes minimal spherical surfaces immersed in $S^3$, and that some 
solutions to our system determine immersed surfaces in $S^3$ whose Gaussian curvatures change sign periodically and whose mean curvature is a periodic nonzero function. 

\medskip

We include an analysis of travelling wave solutions for several reasons. First, they provide a natural finite-dimensional reduction of this system to 
autonomous second-order ODEs, making it possible to study dynamics by means of phase-plane methods and the 
classical energy integral of one-dimensional conservative mechanics \cite{Arnold}. 
Second, in the integrable setting, travelling waves are coherent structures: they provide an explicit class of solutions 
through which we can test how the integrability of the generalized Konno--Oono system is reflected in 
concrete wave profiles \cite{KO,DrazinJohnson}. Third, and most importantly in a geometric context, these
solutions allow us to track directly how the geometry of the associated surfaces in $S^3$ evolves along a distinguished symmetry reduction.
Indeed, periodic oscillations of $R(\xi)=r(x- vt)$, see Section 5, determine the family of immersed surfaces in $S^3$ mentioned in the previous paragraph; 
in addition, in the small-amplitude limit, the periodic orbit collapses to the constant-amplitude travelling wave
$R(\xi)\equiv R_*$, whose associated immersed surface has identically vanishing Gaussian curvature and nonzero constant mean curvature.

\medskip

The first part of this paper (essentially Section 2) is partially based on the contribution of the second named author 
to the G\"urses-Fest Conference (Bilkent University, December 2025) honouring Prof. Metin G\"urses, see \cite{Reyes-gurses}. In this section we  
review the classical theory of equations of pseudo-spherical type, emphasizing the geometric mechanisms that produce zero-curvature representations, conservation laws, pseudo-potentials, and 
generalized B\"acklund transformations. This motivating review allows us to turn, in Section~3, to systems of Lund--Regge type. 
We recall their main characteristics ---zero-curvature representations and formal production of conservation laws--- and relate this class of equations to the Fokas--Gelfand viewpoint
\cite{FokasGelfand1996}. In Section~4, we introduce the generalized Konno--Oono system, which will serve as the main example of our
Lund--Regge framework. We derive an associated $\mathfrak{so}(3,\mathbb{R})$-valued linear problem 
for (KO), and prove that the generalized Konno--Oono system admits infinitely many non-trivial local conservation laws. 
Sections~5 and~6 are devoted to the travelling wave reduction and its consequences. We consider a distinguished subclass of travelling waves and show that the reduced dynamics is governed by a 
one-dimensional conservative system; we also identify a family of periodic orbits through phase-plane analysis, and relate the corresponding solutions to the immersed surfaces in $S^3$
that they describe. Finally, in Section 7, we sketch a numerical scheme that permits us to compute travelling wave solutions of the generalized Konno--Oono system and to illustrate our theoretical results.

\section{Equations of pseudo-spherical type}

\subsection{The main definitions}

\begin{definition}
A two-dimensional manifold $M$ is called a \emph{pseudo-spherical surface} if
there exist one-forms $\omega^1,\omega^2,\omega^3$ on $M$ such that
$\omega^1\wedge\omega^2\neq 0$ and

\begin{equation}\label{eq:pss-structure}
  d\omega^1=\omega^3\wedge\omega^2,\qquad
  d\omega^2=\omega^1\wedge\omega^3,\qquad
  d\omega^3=\omega^1\wedge\omega^2.
\end{equation}
\end{definition}

If $M$ is pseudo-spherical, the metric $ds^2=(\omega^1)^2+(\omega^2)^2$ 
has constant Gaussian curvature $K=-1$, and $\omega_{12} = \omega^3$ is the torsion-free
metric connection one-form.  We also note that if the last equation in $(\ref{eq:pss-structure})$ is replaced by $d\omega^3=-\omega^1\wedge\omega^2$, then
$M$ is a \emph{spherical surface}. In this case, $K=+1$.

Let us consider a partial differential equation of arbitrary order in two independent variables, 
\begin{equation}\label{eq:PDE-general}
  \Xi(x,t,u,u_x,\dots,u_{x\cdots t\cdots})=0.
\end{equation}

\begin{definition} \label{psseq}
Equation \eqref{eq:PDE-general} is of \emph{pseudo-spherical type} (or, it describes pseudo-spherical surfaces) if 
there exist one-forms $\omega^\alpha$ $(\alpha=1,2,3)$ of the form
\begin{equation}\label{eq:omega-alpha}
  \omega^\alpha = f_{\alpha 1}(x,t,u,\dots)\,dx + f_{\alpha 2} (x,t,u,\dots)\,dt,
\end{equation}
in which each coefficient depends on $x,t,u$ and finitely many derivatives of $u$,
such that for every solution $u=u(x,t)$ of \eqref{eq:PDE-general}, the pulled-back
forms $\omega^\alpha(u(x,t))$ satisfy \eqref{eq:pss-structure}.
\end{definition}

When $u(x,t)$ is {\em generic}, this is, $(\omega^1\wedge\omega^2)(u(x,t))\neq 0$, the domain of $u(x,t)$ is 
endowed with a Riemannian metric of constant curvature $K=-1$.

\smallskip

The above definition is intrinsic:  it does not automatically produce extrinsic connection forms so that a full Gauss--Codazzi system holds when 
pulled back to solutions. A systematic relation between the Chern--Tenenblat viewpoint and the
Gauss--Codazzi viewpoint has been investigated only recently, as we explain below.

Definition 2 is due to Chern and Tenenblat \cite{CT}, with essential input coming from Sasaki,
see \cite{S}, as we anticipated in Section 1. Many
integrable equations fit into this class, including sine-Gordon, mKdV and (higher order)
KdV (\cite{CT} and the later paper \cite{R-hierarchy}), Camassa--Holm \cite{R-ch}, Hunter--Saxton \cite{R-hs}, 
Kaup--Kupershmidt and Sawada--Kotera (\cite{R-selecta} and the later paper \cite{neto1}),
and all integrable cases of the Rosenau--Hyman compacton equations \cite{BDDR}. In addition, many interesting 
classification results have been obtained\footnote{Instead of Definition \ref{psseq}, classification results use the stronger condition that an equation is of pseudo-spherical type 
if it is necessary and sufficient for Equations \eqref{eq:pss-structure} to hold, see \cite{CT,T}.}, see for instance the papers \cite{CCT,CT,DT,KT,KeT,KeT2,neto1,Ra} by 
Tenenblat and her collaborators and
students, and it seems noteworthy to observe that some equations have appeared in these classifications long 
before proving themselves of importance for Mathematical Physics. An example of this remark is provided by 
the short pulse equation, derived in 2004 by Sch\"afer and Wayne, see \cite{SW},
\begin{equation} \label{eq:ra1}
u_{xt} = u + \frac{1}{6} (u^3)_{xx}\; .
\end{equation}
This equation is a particular case of a class of equations of pseudo-spherical type studied by Rabelo in 1989 (!), see \cite[Equation (2)]{Ra}, and \cite{T} for explicit computations. 
It seems to us that truly interesting phenomena in nonlinear 
science remain to be discovered among the equations classified in \cite{CCT,CT,DT,KT,KeT,KeT2,neto1,Ra} and other papers 
by Tenenblat and her research group. See, for instance, the travelling wave solutions to 
(\ref{eq:ra1}) appearing in \cite{SS}.

\begin{remark}
An important point is the following: assume that $\Xi=0$ is of 
pseudo-spherical type with associated one-forms $\omega^\alpha = f_{\alpha 1} dx + f_{\alpha 2} dt$. The functions
$f_{\alpha \beta}$ are smooth functions depending on arbitrary (but finite) numbers of derivatives of a solution
$u(x,t)$. How can we be sure that these derivatives exist? This question 
is intimately related to regularity properties of the solutions to $\Xi=0$, and to answer it requires  
careful analytic handling. Interesting advances along these lines have been made by I.L. Freire, see his  
recent survey  {\rm \cite{IF}}. 
\end{remark}

\subsection{Why pseudo-spherical type is useful}
It is, of course, interesting from a geometric point of view to find out that solutions to a given equation 
can be used to construct pseudo-spherical metrics.  On the other hand, from a more applied 
perspective, this framework allows us to (we assume that $\Xi=0$ is an equation of pseudo-spherical type with associated one-forms $\omega^\alpha$):

\begin{itemize}
\item construct integrable overdetermined $\mathfrak{sl}(2,\mathbb{R})$-valued linear problems (see \cite{CT,T}). 
Indeed, the linear problem
$$
\left( \begin{array}{c} d \psi_1 \\ d \psi_2 \end{array} \right) = \frac{1}{2} \left( \begin{array}{cc}
\omega^2 & \omega^1 - \omega^3 \\ \omega^1 + \omega^3 & - \omega^2 \end{array} \right) \left( \begin{array}{c}  
\psi_1 \\ \psi_2 \end{array} \right)
$$
is integrable whenever $\Xi=0$ holds. 

\item derive infinite families of conservation laws (see \cite{CT, T, Reyes2000, R-ch}). 
It is proven in \cite{CT,T} (see also a 1988 paper by Cavalcante and Tenenblat cited in \cite{T,TT}) that the Pfaffian system
\begin{equation} \label{t2}
\omega^3 - d \phi + \sin\phi\, \omega^1 + \cos\phi \omega^2 = 0
\end{equation}
is completely integrable for $\phi$ whenever $\Xi=0$ holds, and in this case, the one-form 
\begin{equation}  \label{t3}
\widehat{\omega} = \cos\phi \omega^1 - \sin\phi \omega^2
\end{equation}
is closed on solutions to $\Xi=0$. Thus, if the one-forms $\omega^\alpha$ depend on a parameter $\eta$, say, then $\phi$ (and therefore $\omega$) can in principle be expanded as power series in
$\eta$, and we would obtain, at least formally, an infinite sequence of conservation laws. This observation has 
been recently revisited by Tenenblat and Tumpach \cite{TT}.  In the variational bicomplex language, if such conservation laws 
are local (something not automatically guaranteed, see \cite{CT,T}) they determine horizontal cohomology classes 
of the equation manifold; proving that these classes are non-zero is a separate and often delicate issue.

\item build quadratic pseudo-potentials and linearizing transformations (see \cite{Reyes2000, R-ch} and examples appearing in \cite{Reyes-gurses}). 
Indeed, the Pfaffian system (\ref{t2}) transforms into either
\begin{equation}
-2 \, d \, \Gamma = \sigma^{3} + \sigma^{2} - 2 \, \Gamma \, \sigma^{1} + \Gamma^{2} \, 
(\sigma^{3} - \sigma^{2})          \; ,                                            \label{pseudo1}
\end{equation}
or
\begin{equation}
2 \, d \, {\gamma} = \sigma^{3} - \sigma^{2} - 2 \, {\gamma} \, \sigma^{1} + {\gamma}^{2} \,
(\sigma^{3} + \sigma^{2})                                             	\label{pseudo2}
\end{equation}
under trigonometric changes of variables, and consequently, the closed one-form $\widehat{\omega}$ becomes either
\begin{equation}
{\sigma}^{1} - \Gamma \, ({\sigma}^{3} - {\sigma}^{2})  \; ,  \qquad
\mbox{ or } \qquad 
 - {\sigma}^{1} + {\gamma} \, ({\sigma}^{3} +{\sigma}^{2}) \;    \label{alla2}
\end{equation}
respectively. The functions $\Gamma$ and $\gamma$ are quadratic pseudo-potentials for $\Xi=0$. The forms in
(\ref{alla2}) were used in \cite{R-ch} to prove for the first time that the Camassa-Holm equation has an infinite number 
of non-trivial local conservation laws. 

\item obtain nonlocal symmetries (see \cite{R-ch,R-selecta}). 
Nonlocal symmetries are usually related to integrals of dependent variables.  
In the case of equations of pseudo-spherical type, we can seek nonlocal
symmetries depending, for instance, on the pseudo-potential $\Gamma$ appearing in (\ref{pseudo1}) . If $\Gamma$ depends on
a parameter, we can find an infinite number of (nonlocal) symmetries, or even a recursion operator. The
reader is referred to \cite{R-ch,R-selecta} for details.

\item provide a geometric understanding of hierarchies of integrable equations (see \cite{R-hierarchy}). 
It is a classical observation that scalar integrable equations ``come in hierarchies", that is, an integrable equation
such as the Korteweg-de Vries equation determines an infinite number of higher order equations that are, in a 
rigorous sense, commuting flows. This observation has a geometric counterpart that is investigated in 
\cite{R-hierarchy}, partially motivated by \cite{CP}.

\item construct generalized B\"acklund transformations (see \cite{KT, R-hierarchy}).
Pseudo-spherical surfaces are all locally isometric (a 
discussion of this fact within the framework of equations of pseudo-spherical type is in \cite{KT}). Assume that 
we have two equations of pseudo-spherical type, $\Xi_1=0$ and $\Xi_2=0$, and generic solutions $u_1$ and $u_2$ to 
$\Xi_1=0$ and $\Xi_2=0$ respectively. Then, $u_1$ and $u_2$ determine locally isometric 
metrics with Gaussian curvature $K=-1$. It is possible to ``unravel" this isometry and
{\em determine} $u_2$ starting from $u_1$! This is proven in \cite{KT} (and later in \cite{R-hierarchy} in the context of {\em hierarchies} of equations of pseudo-spherical type). This 
geometrically induced ``unravelling" is what we call a generalized B\"acklund transform. 
\end{itemize}

\subsection{The immersion problem}
Assume that $\Xi=0$ is an equation of pseudo-spherical type with 
associated one-forms $\omega^\alpha$. Can we find extrinsic connection one-forms
$\omega_{13}, \omega_{23}$ depending on $x,t,u$ and finitely many derivatives of $u$ so that
the Gauss--Codazzi equations are satisfied upon pullback to (generic) solutions?

For $k$th-order evolution equations $u_t=F$ describing pseudo-spherical surfaces
with $\omega^\alpha=f^\alpha_1 dx + f^\alpha_2 dt$ and $f^2_1=\eta$, a difficult result
obtained by Kahouadji, Kamran, and Tenenblat, see \cite{KKT2019}, gives severe obstructions to the
existence of such extrinsic data:
\begin{theorem} 
Except for $k$-th order evolution equations of the form
\begin{equation} \label{(10)}
\frac{\partial u}{\partial t}
=
\frac{1}{f_{11,u}}
\left(
\sum_{i=0}^{k-1}
f_{12,\partial_i u/\partial x^i}
\cdot
\frac{\partial^{i+1} u}{\partial x^{i+1}}
+
(\beta f_{11} - \eta f_{12})
\right),
\qquad k \ge 2,
\end{equation}
where $f_{11,u} \neq 0$, $\eta \in \mathbb{R}$, and 
$f_{12,\partial^{k-1} u/\partial x^{k-1}} \neq 0$, 
there exists no $k$-th order evolution equation of order $k \ge 2$ 
describing $\eta$ pseudo-spherical surfaces, with associated one-forms $\omega^\alpha = f_{\alpha 1} dx +
f_{\alpha 2} dt$  satisfying $f_{21} = \eta$, 
with the property that the coefficients of the second fundamental forms of the 
local isometric immersions of the surfaces associated to the solutions $u$ 
of the equation depend on a jet of finite order of $u$. 

Moreover, the coefficients of the second fundamental forms of the local 
isometric immersions of the surfaces determined by the solutions $u$ of {\rm (\ref{(10)})}
are universal, i.e., they are universal functions of $\eta x + \beta t$, 
independent of $u$.
\end{theorem}

Thus, for {\em many} equations describing pseudo-spherical surfaces 
with associated one-forms $\omega^\alpha$, we cannot expect to find extrinsic 
connection one-forms $\omega_{13}, \omega_{23}$ depending on finite order jets of $u$ such that they, together 
with the forms $\omega^\alpha$, satisfy the structure equations of a surface locally immersed in $\mathbb{R}^3$. 
The $K = -1$ condition is an important constraint: as we will see in Section 3, {\em it is} possible to find 
local immersion data if we do not fix the Gaussian curvature (or if $K = +1$) or, if we consider surfaces immersed in a curved ambient space.

\begin{remark}
There exists a complementary approach to the immersion problem, going back
to Sym, {\rm \cite{Sym}}: starting from a parameter-dependent
zero-curvature representation for $\Xi=0$, solve this PDE together with the associated
linear problem, and then build an immersed surface from these solutions.
This philosophy is explained in the review paper {\rm \cite{GT}} by G\"urses and Tek; an important earlier reference is 
Fokas--Gelfand {\rm \cite{FokasGelfand1996}}. This last paper will be used in Theorem $\ref{fg}$ below.
\end{remark}

\section{Lund--Regge type systems and surfaces immersed in $S^3$}

We can ask whether the
Chern--Tenenblat approach (with or without immersions) admits generalizations. There are several roads we can
take:
\begin{itemize}
\item  We could try an $n$-dimensional version of Definition 2. This looks quite difficult, actually; Tenenblat 
and her coworkers have found explicit examples of $n$-dimensional integrable equations that are analogous to some
equations of pseudo-spherical type ({\em e.g.} they have considered an $n$ dimensional sine-Gordon equation, see 
\cite{T,TT} and references therein) but no general theory has been developed, to the best of our knowledge. 
\item We could try to change the $K=-1$ condition: this approach has been explored in the papers \cite{DT,KeT,KeT2}, 
in which the authors consider (systems of) equations that describe {\em spherical} surfaces. 
\item We could revisit the immersion approach used in \cite{GN,L}. This approach seems
promising and is currently under development, see \cite{BHR}. The main idea here, also partially motivated by 
the immersion theory of equations of pseudo-spherical type, is to introduce a class of equations describing 
surfaces which are immersed in a {\em curved} ambient space.
\end{itemize}

Our starting point is the Lund--Regge system in independent variables $\sigma,\tau$, 
\begin{align}
  \theta_{\tau\tau}-\theta_{\sigma\sigma} - \cos\theta\sin\theta
  +\frac{\cos\theta}{\sin^3\theta}\Big(\lambda_\tau^2-\lambda_\sigma^2\Big) &= 0,
  \label{eq:LR1}\\
  \partial_\tau\Big(\cot^2\theta\,\lambda_\tau\Big)
  -\partial_\sigma\Big(\cot^2\theta\,\lambda_\sigma\Big) &= 0 \; ,
  \label{eq:LR2}
\end{align}
introduced in \cite{L}. We define
\[
  \omega^1 = \cos\theta\, d\sigma,\qquad \omega^2 = \sin\theta\, d\tau,
\]
and connection one-forms
\begin{align*}
  \omega_{12} &= \theta_\tau\,d\sigma + \theta_\sigma\,d\tau,\\
  \omega_{13} &= \frac{\lambda_\sigma}{\sin\theta}\,d\sigma
              + \frac{\lambda_\tau}{\sin\theta}\,d\tau,\\
  \omega_{23} &= \frac{\lambda_\tau\cos\theta}{\sin^2\theta}\,d\sigma
              + \frac{\lambda_\sigma\cos\theta}{\sin^2\theta}\,d\tau.
\end{align*}
Then, a solution $(\theta,\lambda)$ of \eqref{eq:LR1}--\eqref{eq:LR2} with
$(\omega^1\wedge\omega^2)(\theta(\sigma,\tau),\lambda(\sigma,\tau))\neq 0$ determines a surface immersed in $S^3$ 
equipped with moving coframe $\omega^1(\theta(\sigma,\tau),\lambda(\sigma,\tau))$ and 
$\omega^2(\theta(\sigma,\tau),\lambda(\sigma,\tau))$ and connection one-forms as above; in fact, the five
one-forms $\omega^1, \cdots, \omega_{23}$ satisfy the structure equations of a surface locally immersed in $S^3$
if and only if $(\theta,\lambda)$ solves \eqref{eq:LR1}--\eqref{eq:LR2}. 

\smallskip

It is therefore very natural to make the following definition, after \cite{BHR}:

\begin{definition} \label{elrt}
A (system of) differential equation(s)
$\Xi(\sigma,\tau,u,v,\dots)=0$ in dependent variables $u,v$ is of
\emph{Lund--Regge type} if there exist functions $f_{ij}$ and $g_{ikj}$
(with $g_{ikj}=-g_{kij}$, $i,k\in\{1,2,3\}$, $i\neq k$, $j\in\{1,2\}$),
depending on $\sigma,\tau,u,v$ and finitely many derivatives of $u, v$, such that the
one-forms
\[
  \omega^i = f_{i1}\,d\sigma + f_{i2}\,d\tau,\qquad
  \omega_{ik} = g_{ik1}\,d\sigma + g_{ik2}\,d\tau,\qquad (i\neq k)
\]
satisfy the structure equations
\begin{align*}
  d\omega^1 &= \omega_{12}\wedge\omega^2,\qquad
  d\omega^2 = \omega^1\wedge\omega_{12},\qquad
  0 = \omega^1\wedge\omega_{13} + \omega^2\wedge\omega_{23},\\
  d\omega_{12} &= \omega_{23}\wedge\omega_{13} - \omega^1\wedge\omega^2,\qquad
  d\omega_{13} = \omega_{12}\wedge\omega_{23},\qquad
  d\omega_{23} = \omega_{13}\wedge\omega_{12},
\end{align*}
of a surface immersed in $S^3$, 
whenever $(u(\sigma,\tau),v(\sigma,\tau))$ is a solution of $\Xi=0$.
\end{definition}

Of course, we can also consider systems of equations with more than two dependent variables.  If $\Xi=0$
is of Lund--Regge type with associated one-forms $\omega^1, \omega^2, \omega_{12}, \omega_{13}, \omega_{23}$, and 
we consider a solution $s = (u(\sigma,\tau), v(\sigma,\tau) )$ such that the pull-back of 
$\omega^1 \wedge \omega^2$ by $s$ is different from zero, then the pull-back by $s$ of these five one-forms 
determine a surface that is locally immersed in $S^3$. 

\medskip

The first important property of an equation describing Lund--Regge surfaces, is that it is the 
integrability condition of an $\mathfrak{so}(3,\mathbb{R})$-valued linear system. We quote from \cite{BHR}:
\begin{theorem}
Let $\Xi=0$ be a system of differential equations in two independent variables $\sigma,
\tau$, and suppose that there exist one-forms
$\omega^1, \omega^2, \omega_{12}, \omega_{13}, \omega_{23}$ depending on $\sigma,\tau$ and on
a finite number of derivatives of the dependent variables appearing in $\Xi=0$, such that the equations
\begin{equation} \label{zcceq}
d \omega^1 = \omega_{12} \wedge \omega^2\; , \quad \quad d \omega^2 = \omega^{1} \wedge \omega_{12}\; ,
\quad \quad \omega^1 \wedge \omega_{13} + \omega^2 \wedge \omega_{23} = 0
\end{equation}
hold on solutions. Then, this equation is of Lund-Regge type if and only if the matrix
\begin{equation} \label{matrix_omega}
\Omega = \left( 
\begin{array}{ccc} 0 & \omega_{12} & \omega_{13} + \omega^2 \\
	    -\omega_{12} & 0 & \omega_{23} - \omega^1 \\
	    -\omega_{13} - \omega^2 & - \omega_{23} + \omega^1 & 0
\end{array}
        \right)
\end{equation}
satisfies the zero curvature equation 
\begin{equation} \label{zccc}
d\, \Omega = \Omega \wedge \Omega
\end{equation}
on solutions to $\Xi=0$.
\end{theorem}

The standard identification between $\mathfrak{su}(2)$ and $\mathfrak{so}(3)$ allows us to connect this theory with the theory of
equations describing spherical surfaces. We find that an equation of Lund-Regge 
type with associated one-forms as in Definition $\ref{elrt}$ describes spherical surfaces with associated 
one-forms
$$
\overline{\omega}_1 = - \omega_{13} - \omega^2\; , \quad \quad 
\overline{\omega}_2 = \omega_{12}\; , \quad \quad 
\overline{\omega}_3 = \omega_{23} - \omega^1 \; .
$$
We have not stated rigorously the notion of an equation of spherical type, but the correct definition can be
easily obtained from Definition \ref{psseq} and Remark 1. It appears explicitly in \cite{DT}.  
Thus, if $(u,v)$ is a solution to an equation of Lund-Regge type, and the forms $\overline{\omega}_i$ are defined 
as above, we obtain a Riemannian metric of Gaussian curvature $K=+1$ on an open set $V \subset \mathbb{R}^2$ 
contained in the domain of $(u,v)$, as long as
$\overline{\omega}_1(u,v) \wedge \overline{\omega}_2(u,v) = (- \omega_{13}(u,v) - \omega^2(u,v) )\wedge 
\omega_{12}(u,v) \neq 0$ on $V$.

\smallskip

The importance of this rather elementary remark is that it implies, essentially
after \cite{Reyes2000}:

\begin{proposition} \label{l1}
Let $\Xi =0$ be an equation of Lund-Regge type with associated one-forms as in Definition $\ref{elrt}$.
The Pfaffian system 
\begin{eqnarray}
- 2 d \Gamma & = & 
(\omega_{13} + \omega^2 + i \omega_{12}) - 2 i \Gamma 
(\omega_{23} - \omega^1) + \Gamma^{2}(\omega_{13} + \omega^2 - i \omega_{12})           \label{sgammalr} 
\end{eqnarray}
is completely integrable on solutions to $\Xi =0$, and the one-form
\begin{eqnarray}
\Theta & = & i (\omega_{23} - \omega^1) - \Gamma (\omega_{13} + \omega^2 - i \omega_{12})	\label{alla1lr}
\end{eqnarray}
is closed on solutions to $\Xi =0$. Analogous claims hold for the Pfaffian system
\begin{eqnarray}
2 d \hat{\Gamma} & = & 
(\omega_{13} + \omega^2 - i \omega_{12}) - 2 i \hat{\Gamma} 
(\omega_{23} - \omega^1) + \hat{\Gamma}^{2}(\omega_{13} + \omega^2 + i \omega_{12})       \label{sgamma2lr}           
\end{eqnarray}
and the one-form
\begin{eqnarray}
\hat{\Theta} & = & -i (\omega_{23}-\omega^1)+\hat{\Gamma} (\omega_{13}+\omega^2 + i \omega_{12}) \;\label{alla2lr}
\end{eqnarray} 
\end{proposition}

Thus, equations of Lund-Regge type admit quadratic pseudo-potentials 
and in principle, if the associated one-forms depend on a parameter, infinitely many conservation laws.

\smallskip

We now present a relation between the foregoing theory and 
\cite{Sym,FokasGelfand1996}. We assume familiarity with \cite{O}. Let us suppose that $\Xi=0$ is of Lund-Regge 
type, and that $v_Q$ is a generalized symmetry of $\Xi=0$ in evolutionary form. We write the matrix $\Omega$
given by (\ref{matrix_omega}) as an $\mathfrak{su}(2)$-valued matrix of differential one-forms, and we recall that (the 
infinite prolongation of) $v_Q$ acts on differential forms via Lie derivatives. Thus, we can define
$$
\Lambda = L_{v_Q} \Omega = A d\sigma + B d \tau
$$
for $\mathfrak{su}(2)$-valued matrices $A,B$ whose entries are differential functions. Since $\Omega$ satisfies (\ref{zccc}), 
we find, by computing Lie derivatives, that on solutions to $\Xi =0$ the equation 
\begin{equation} \label{eq:4}
A_\tau - B_\sigma + [A , T] + [\Sigma , B] = 0\; ,
\end{equation}
holds, in which we have set $\Omega = \Sigma \, d \sigma + T \, d \tau$. We see that 
Equation (\ref{eq:4}) is precisely 
Equation (2.4) in \cite{FokasGelfand1996}, identifying $(u , v)$ of that reference with our $(\sigma , \tau)$.
Equations (\ref{zccc}) and (\ref{eq:4}) imply that we can apply \cite[Theorem 2.1]{FokasGelfand1996}. We obtain:

\begin{theorem} \label{fg}
Let $\Xi=0$ be of Lund-Regge type, and let $v_Q$ be a generalized symmetry of $\Xi=0$ in evolutionary form.
We consider matrices $\Omega$ and $\Lambda$ as above, and we assume the non-degeneracy conditions $[A,B] \neq 0$ and $A \neq 0$. Then, $\Xi=0$ describes surfaces immersed in a 
three-dimensional flat space with first fundamental form
\begin{equation} \label{iff}
\left< A , A \right> d\sigma^2 + 2 \left< A , B \right> d \sigma \, d \tau + \left< B , B \right> d \tau^2\; ,
\end{equation}
and second fundamental form 
\begin{equation} \label{iiff}
\left< A_\sigma + [A,\Sigma], C \right> d\sigma^2 + 2 \left< A_\tau+[A,T] , C \right> d \sigma \, d \tau + 
\left< B_\tau + [B,T] , C \right> d \tau^2\; ,
\end{equation}
in which $\left< X , Y \right> = -(1/2)\, {\rm trace}(X Y)$ for $X,Y \in su(2)$, and 
$C =[A,B]/\sqrt{\left<[A,B],[A,B]\right>}\,$.
\end{theorem}

It may be of interest to present explicit formulas for the coframe and connection one-forms implicit in
(\ref{iff}) and (\ref{iiff}). We write the $\mathfrak{su}(2)$-matrix 
$$
\Omega=\frac{1}{2}
\left(
\begin{array}{cc}
i\, \overline{\omega}_2 & \overline{\omega}_1+i\, \overline{\omega}_3  \\
-\overline{\omega}_1+i\, \overline{\omega}_3 & -i\, \overline{\omega}_2  
\end{array} 
\right)
$$
as $\Omega = \Sigma d \sigma + T d \tau$, with 
\begin{equation} \label{omatrix}
\Sigma = \frac{1}{2} \left( \begin{array}{cc}
i\, \overline{a}_2 & \overline{a}_1+i\, \overline{a}_3  \\
-\overline{a}_1+i\, \overline{a}_3 & -i\, \overline{a}_2  
\end{array} \right) \quad \mbox{ and } \quad 
T = \frac{1}{2} \left( \begin{array}{cc}
i\, \overline{b}_2 & \overline{b}_1+i\, \overline{b}_3  \\
-\overline{b}_1+i\, \overline{b}_3 & -i\, \overline{b}_2  
\end{array} \right) \; ,
\end{equation}
and we set $\Lambda = A d \sigma + B d \tau$ with
\begin{equation} \label{lmatrix}
A = \frac{1}{2} \left( \begin{array}{cc}
i\, \overline{\alpha}_2 & \overline{\alpha}_1+i\, \overline{\alpha}_3  \\
-\overline{\alpha}_1+i\, \overline{\alpha}_3 & -i\, \overline{\alpha}_2  
\end{array} \right)  \quad \mbox{ and } \quad 
B = \frac{1}{2} \left( \begin{array}{cc}
i\, \overline{\beta}_2 & \overline{\beta}_1+i\, \overline{\beta}_3  \\
-\overline{\beta}_1+i\, \overline{\beta}_3 & -i\, \overline{\beta}_2  
\end{array} \right) \; .
\end{equation}
We also identify $\mathfrak{su}(2)$ with $\mathbb{R}^3$ via the map $\Phi : \mathbb{R}^3 \rightarrow \mathfrak{su}(2)$ given by 
\[
\Phi(x_1,x_2,x_3)=
\frac12
\begin{pmatrix}
i x_2 & x_1+i x_3\\
-x_1+i x_3 & -i x_2
\end{pmatrix} \; ,
\]
so that the product $\left< \cdot , \cdot \right>$ corresponds to the Euclidean dot product, while the Lie bracket corresponds to the opposite of the standard cross product: 
\[
\langle \Phi(x),\Phi(y)\rangle=\frac14 x\cdot y,
\qquad \mbox{ and } \qquad
[\Phi(x),\Phi(y)]= - \Phi(x\times y)\; .
\]
We write $a=(\overline{a}_1,\overline{a}_2,\overline{a}_3)$ and define $b, \alpha, \beta$ similarly. 
In this notation, (\ref{iff}) and (\ref{iiff}) become
\[
I=
\frac14|\alpha|^2\,d\sigma^2
+\frac12(\alpha\cdot\beta)\,d\sigma\,d\tau
+\frac14|\beta|^2\,d\tau^2 
\]
and
\[
II=
\frac12\big((\alpha_\sigma-\alpha\times a)\cdot n\big)d\sigma^2
+\big((\alpha_\tau-\alpha\times b)\cdot n\big)d\sigma d\tau
+\frac12\big((\beta_\tau-\beta\times b)\cdot n\big)d\tau^2 \; ,
\]
in which $n=-\alpha\times\beta/|\alpha\times\beta|$. We also introduce the covariant derivatives\footnote{We note 
that using these definitions, the compatibility equations $d\Omega = \Omega \wedge \Omega$ and (\ref{eq:4}) 
become, simply,   $a_{\tau}-b_{\sigma}=a\times b$ and $D_{\tau}\alpha = D_{\sigma}\beta$ respectively.}
\[
D_\sigma X=X_\sigma-X\times a,
\qquad \mbox{ and } \qquad
D_\tau X=X_\tau-X\times b\; 
\]
and we set $\triangle =|\alpha\times\beta|^2 = |\alpha|^2|\beta|^2-(\alpha\cdot\beta)^2$. 
The moving coframe and connection one-forms associated with the fundamental forms (\ref{iff}) and (\ref{iiff}) are 
\[
\omega^1=\frac12|\alpha|\,d\sigma
+\frac12\frac{\alpha\cdot\beta}{|\alpha|}\,d\tau 
\qquad \mbox{ and } \qquad
\omega^2=\frac12\frac{\sqrt{\triangle}}{|\alpha|}\,d\tau\; ,
\]
and 
\begin{eqnarray*}
\omega_{12}
& = &
-\frac{1}{|\alpha|^2\sqrt{\triangle}}
\Big(
|\alpha|^2(D_\sigma\alpha\cdot\beta)
-(\alpha\cdot\beta)(\alpha_\sigma\cdot\alpha)
\Big)\,d\sigma 
-\frac{1}{|\alpha|^2\sqrt{\triangle}}
\Big(
|\alpha|^2(D_\tau\alpha\cdot\beta)
-(\alpha\cdot\beta)(\alpha_\tau\cdot\alpha)
\Big)\,d\tau\; , \\\\[1mm]
\omega_{13}
& = &
\frac{1}{|\alpha|}\,(D_\sigma\alpha\cdot n)\,d\sigma
+
\frac{1}{|\alpha|}\,(D_\tau\alpha\cdot n)\,d\tau \; , \\\\[1mm]
\omega_{23}
& = &
\frac{1}{|\alpha|\sqrt{\triangle}}
\Big(
|\alpha|^2(D_\sigma\beta\cdot n)
-(\alpha\cdot\beta)(D_\sigma\alpha\cdot n)
\Big)\,d\sigma
+
\frac{1}{|\alpha|\sqrt{\triangle}}
\Big(
|\alpha|^2(D_\tau\beta\cdot n)
-(\alpha\cdot\beta)(D_\tau\alpha\cdot n)
\Big)\,d\tau \; .
\end{eqnarray*}

A quick corollary of these computations, that shows how different from the pseudo-spherical case discussed in 
Subsection 2.3 the $K=+1$ case is, is the following:

\begin{corollary}
If $\Xi=0$ describes spherical surfaces with associated one-forms $\omega^\alpha$, it is an equation of
Lund-Regge type that describes totally geodesic (hence minimal) spherical surfaces immersed in 
$S^3$. In addition, if $\Xi=0$ admits a (generalized) symmetry, we can find one-forms 
$\omega_{13}$ and $\omega_{23}$ depending on finite order jets of $u$, such that the five one-forms
$\omega^\alpha$ ($\alpha=1,2,3$), $\omega_{13}$, $\omega_{23}$ satisfy the structure 
equations of a spherical surface immersed in Euclidean 3-space on solutions to $\Xi=0$. 
\end{corollary}
The proof follows from the above computations and \cite[Example 3]{BHR}.

\section{The generalized Konno--Oono system}

In this section we investigate in detail an example of an equation of Lund-Regge type.  
Konno and Oono introduced the well-known coupled system
\begin{equation} \label{sko}
  q_t + 2 r r_x = 0\; ,\qquad r_{xt}-2qr=0 \; ,
\end{equation}
in \cite{KO,KO2}. It is not difficult to prove that it describes spherical surfaces with associated one-forms
\begin{equation} \label{ko00}
  \omega^1 = 2\lambda q\,dx - \frac{1}{\lambda}\,dt,\qquad
  \omega^2 = -2r\,dt,\qquad
  \omega_{12} = -2\lambda r_x\,dx,
\end{equation}
where $\lambda\neq 0$ is a real parameter. Following \cite{BHR}, we introduce an auxiliary field $\delta(x,t)$ 
and we set
\begin{equation} \label{ko01}
  \omega_{13} = 4\lambda\,\delta_x r\,dx,\qquad \mbox{ and } \qquad 
  \omega_{23} = -2\delta_x\,dx - 2\delta_t\,dt\; .
\end{equation}
Then, the one-forms $\{\omega^1,\omega^2,\omega_{12},\omega_{13},\omega_{23}\}$
satisfy the structure equations of a surface immersed in $S^3$ if and only if
$r(x,t),q(x,t),\delta(x,t)$ solve
\begin{align}
  q_t + 2 r r_x &= 0, \label{eq:1}\\
  r_{xt} - 2qr - 4r\delta_x\delta_t &= 0, \label{eq:2}\\
  \delta_{xt} r + r_t\delta_x + r_x\delta_t &= 0. \label{eq:3}
\end{align}
We call (\ref{eq:1})--(\ref{eq:3}) the {\em generalized Konno--Oono system}\footnote{We are aware that the term ‘generalized Konno–Oono equation/system’ has also been used in the coupled-dispersionless 
literature for a different integrable system, see \cite{kk2}. We have decided to keep the name, however, in order to be consistent with \cite{BHR}.}. This system of nonlinear partial 
differential equations appeared for the first time in \cite{BHR}, to the best of our knowledge. 

The Gaussian and mean curvatures of the surfaces described by (\ref{eq:1})--(\ref{eq:3}) are
\begin{equation} \label{curvatures}
  K = \frac12\frac{r_{xt}}{rq} = 1 + \frac{2}{q}\,\delta_x\delta_t
  \qquad \mbox{ and } \qquad 
  H = \frac{1}{rq}\Big(r^2\delta_x + \frac{1}{4\lambda^2}\delta_x + \frac12\delta_t q\Big)
\end{equation}
respectively. In particular, if $\delta$ is constant, the system reduces to (\ref{sko}) and, since in this case $K=1$ and $H=0$, we see that 
the original Konno--Oono system describes minimal spherical surfaces locally immersed in $S^3$. On the other hand, in the travelling wave regime, 
with $q(x,t)=Q(\xi)$, $r(x,t)=R(\xi)$, $\delta(x,t)=\Delta(\xi)$ and $\xi=x-vt$, the geometry of the immersed surface is 
transported rigidly with the wave profile.  Let us make a short advance of the kind of phenomena we will discuss in Sections 5 and 6.
For $v\neq 0$ and on any interval on which $R(\xi)\neq 0$, \eqref{eq:3} and \eqref{eq:1} imply that the Gaussian curvature can be rewritten as
\begin{equation} \label{twgc}
K(\xi)=1-\frac{2v^{2}C^{2}}{R^{4}(\xi)(R^2(\xi)-C_0)}\; 
\end{equation}
for constant numbers $C_0$ and $C$. This shows that surfaces described by the generalized Konno--Oono system may have singularities (if $R^2(\xi) = C_0$) or, if 
$R^2(\xi) \neq C_0$, that the sign of $K$ is completely determined by comparison between $R^4(\xi)(R^2(\xi)-C_0)$ and the 
threshold $2v^{2}C^{2}$. More precisely,
\[
K(\xi)>0 \iff R^{4}(\xi)(R^2(\xi)-C_0) >2v^{2}C^{2} \; . 
\]
In particular, we anticipate that if $R(\xi)$ is oscillatory and $R^{4}(\xi)(R^2(\xi)-C_0)$  crosses the critical level $2v^{2}C^{2}$, the Gaussian curvature will change sign along the travelling wave profile.

\smallskip

Now we focus on structural characteristics  of (\ref{eq:1})--(\ref{eq:3}). 

First, we observe that, since the one-forms (\ref{ko00}) and (\ref{ko01}) depend on the parameter 
$\lambda$, Theorem 3 suggests that the system  (\ref{eq:1})--(\ref{eq:3})  may be
solvable via scattering/inverse scattering. Explicitly, the associated linear system is 
$\psi_x = X\, \psi$, $\psi_t = T\, \psi$, 
in which  
\begin{equation*} 
X = \left( \begin{array}{ccc}
0 & -2 \lambda r_x & 4 \lambda \delta_x r \\
2 \lambda r_x & 0 & -2 \delta_x - 2 \lambda q \\
-4 \lambda \delta_x r  & 2 \delta_x + 2 \lambda q & 0
\end{array}
        \right)\; ,
\quad
T = \left( \begin{array}{ccc}
0 & 0 & - 2 r  \\
0 & 0 & \displaystyle -2 \delta_t + \frac{1}{\lambda}\\
 2 r &\displaystyle 2 \delta_t - \frac{1}{\lambda} & 0
\end{array}
        \right)\; .
\end{equation*}
Matrices of this type are closely related to the geometry of moving frames, curve motions, and spin systems, see for instance \cite{My}. Motivated by
this observation, we record in Proposition~\ref{spin} an explicit spin-type reformulation of a distinguished conserved quantity of the generalized Konno--Oono system: we show that the vector $X=(q,r_x,2r\delta_x)$ satisfies a precession equation
\[
X_t=\Omega_X\times X\; ,
\]
and that the squared length
\[
|X|^2=q^2+r_x^2+4r^2\delta_x^2
\]
is independent of $t$ on solutions. A more complete study of the associated moving-frame,
curve-motion, and spin geometry will be pursued elsewhere.

\smallskip

Second, we consider symmetries. We state the following straightforward proposition.

\begin{proposition}
The algebra of local point symmetries for the generalized Konno--Oono system $(\ref{eq:1})$--$(\ref{eq:3})$ is generated by 
\[
\frac{\partial}{\partial t}\; ,
\qquad
t \frac{\partial}{\partial t} - q \frac{\partial}{\partial q} - r \frac{\partial}{\partial r}\; ,
\qquad
\frac{\partial}{\partial \delta}\; ,  
\]
and the vector fields 
\[
\mathbf X_f = f(x)\frac{\partial}{\partial x} - f_x(x)q\frac{\partial}{\partial q}\; .
\]
Moreover, 
\[
[ \mathbf X_f,\mathbf X_g] = \mathbf X_{fg_x-gf_x}\; .
\]
Thus, the system  $(\ref{eq:1})$--$(\ref{eq:3})$ admits an infinite-dimensional Lie algebra
of local symmetries, isomorphic to the Lie algebra of vector fields in the $x$-variable.
\end{proposition}
\noindent 
We omit the proof since computation of point symmetries has become essentially a matter of symbolic computation. We note that if the $x$-variable belongs to $S^1$, this proposition says that 
the centerless  Virasoro algebra $Vect(S^1)$ is a subalgebra of point symmetries of our generalized Konno--Oono system. 
The presence of this infinite-dimensional local symmetry algebra allows us to give an interesting interpretation to the dependent variables $r, q, \delta$: the variables $r$ and
$\delta$ ``transform as scalar fields'', while $q$ ``transforms as a weight-one density'' under the action of $Diff(S^1)$. Indeed, changing to evolutionary representatives, their variations 
are
\[
\delta_f\, r=-f\, r_x,\qquad
\delta_f\, \delta=-f\, \delta_x,\qquad
\delta_f\, q=-(fq)_x\; .
\]

Third, we turn to the study of conservation laws. Our goal is to prove that our generalized Konno--Oono system
admits infinitely many local conservation laws that are pairwise distinct and non-trivial in horizontal cohomology. We begin with some formal arguments.

\smallskip

Equation (\ref{sgammalr}) becomes the Riccati system 
\begin{equation} \label{sgamma6}
-\Gamma_x = (2 \delta_x r - i r_x ) \lambda + (2 i \delta_x + 2 i q \lambda) \Gamma + 
(2 \delta_x r + i r_x) \lambda \Gamma^2\; 
\end{equation}
and
\begin{equation} \label{sgamma66}
-\Gamma_t = \frac{1}{\lambda}(-  r \lambda + (2 i \lambda \delta_t - i) \Gamma -r  
\lambda \Gamma^2 ) \;  .
\end{equation}
Therefore, if $\displaystyle \Gamma = \sum_{n=0}^\infty \Gamma_n\,\lambda^{-n}$ and we expand (\ref{sgamma6}) in 
powers of $\lambda$, we find
\begin{eqnarray*}
-2 \sum_{n=0}^\infty \Gamma_{n,x}\,\lambda^{-n} & = & (4 \delta_x r - 2 i r_x )\lambda + 
4 i \delta_x \sum_{n=0}^\infty \Gamma_n\,\lambda^{-n} + 4 i q \sum_{n=0}^\infty \Gamma_n\,\lambda^{-n+1} \\
 &  & + 2 (2 \delta_x r + i r_x) \sum_{n, m=0}^\infty \Gamma_n \Gamma_m \,\lambda^{-n-m+1}.
\end{eqnarray*}
Thus, the following equations hold:
\begin{eqnarray}
0 & = & 4 \delta_x r - 2 i r_x + 4 i q \Gamma_0 + 2 (2 \delta_x r + i r_x) \Gamma_0^2 \label{g0} \\
- 2 \Gamma_{0,x} & = &  4 i \delta_x \Gamma_0 + 4 (2 \delta_x r + i r_x) \Gamma_0 \Gamma_1 + 4 i q \Gamma_1  \label{ggg} \\
- 2 \Gamma_{n,x} & = & 4 i \delta_x \Gamma_n + 4 i q \Gamma_{n+1} + 2 (2 \delta_x r + i r_x)\sum_{j=0}^{n+1}
\Gamma_{n-j+1}\Gamma_j \; , \quad n \geq 1\; . \label{gg}
\end{eqnarray}
These equations determine {\em local} expressions for the functions $\Gamma_n$, and therefore Equation 
(\ref{alla1lr}) allows us to conclude that (\ref{eq:1})--(\ref{eq:3}) admits {\em local}  conservation laws. 

\smallskip

We compute $\Gamma_0$ using (\ref{g0}). We set $S = q^2 + 4\delta_x^2 r^2 + r_x^2\,$ and we find
$$
\Gamma_0 = \frac{-q \pm \sqrt{S}}{2 \delta_x r + i r_x}\,i =
\frac{-q \pm \sqrt{S}}{4\delta_x^2 r^2 + r_x^2} (r_x + 2\delta_x r\ i)\; .
$$
The first conservation law arising from (\ref{alla1lr}) is (see also \cite[p. 14]{BHR})
$$
\Theta^{(-1)} = (- 2 i q - (4 \delta_x r + 2 i r_x)\Gamma_0 ) dx \; .
$$
Replacing $\Gamma_0$ into $\Theta^{(-1)}$ and separating real and imaginary parts, we obtain 
$Re(\Theta^{(-1)}) =0$ and
$$
Im(\Theta^{(-1)}) = \mp 2 \sqrt{S} dx\; .
$$
Thus, we find that on solutions to the generalized Konno--Oono system,
\begin{equation}   \label{im}
q^2 + 4\delta_x^2 r^2 + r_x^2 = \kappa(x) \; .
\end{equation}
This conservation law allows us to prove, as anticipated, that there is a connection between the generalized Konno--Oono system and spin structures. 
We make the following observation.

\begin{proposition} \label{spin}
Let us consider the generalized Konno--Oono system $(\ref{eq:1})$--$(\ref{eq:3})$, and define the vector field
\[
{X}(x,t)=(q,\; r_x,\; 2r\delta_x)
\]
and the angular velocity vector
\[
{\Omega}_X(x,t)=(-2\delta_t,\;0,\;2r)\; .
\]
Then, on solutions to the generalized Konno--Oono system, ${X}$ satisfies the precession equation
\[
{X}_t={\Omega}_X \times {X}
\]
on a sphere of radius $\sqrt{\kappa(x)} = \sqrt{q^2+r_x^2+4r^2\delta_x^2}\,$. 
\end{proposition}
\begin{proof}
Define
\[
{X}=(X_1,X_2,X_3)=(q,r_x,2r\delta_x).
\]
From (\ref{eq:1})--(\ref{eq:3}) we obtain
\begin{eqnarray*}
X_{1,t} & = & q_t \; = \; -2rr_x=-2rX_2 \; ;\\
X_{2,t} & = & r_{xt}\; = \; 2qr+4r \delta_x \delta_t =2rX_1+2\delta_tX_3 \; ; \\
X_{3,t} & = & 2 r_t \delta_x + 2r \delta_{xt} =-2\delta_tX_2 .
\end{eqnarray*}
Thus, 
\[
{X}_t=(-2rX_2,\;2rX_1+2\delta_tX_3,\;-2\delta_tX_2) = (-2\delta_t,0,2r) \times X \; ,
\]
that is, ${X}_t={\Omega}_X \times {X}$.  Since (\ref{im}) holds on solutions, we conclude that $|X|^2 = \kappa(x)$ is preserved along solutions to this precession equation.
\end{proof}

\smallskip

Now we go back to the study of conservation laws. Substituting $\Gamma_0$ into (\ref{gg}) and collecting terms, we obtain the explicit formula
\begin{equation} \label{gs}
\Gamma_{n+1}
=
\frac{-\Gamma_{n,x}-2i\delta_x\Gamma_n
-(2r\delta_x+ir_x)\sum_{j=1}^{n}\Gamma_j\Gamma_{n+1-j}}
{\pm 2\,i\sqrt{S}}\; ,
\end{equation}
and therefore if we consider the closed one-form $\Theta$ given by (\ref{alla1lr})  and expand it in powers of $\lambda$, we obtain the conservation laws
\begin{eqnarray*}
\Theta^{(-1)} & = & (-2 i q - 4 \Gamma_0 \delta_x r - 2 i \Gamma_0 r_x  ) dx\; , \\
\Theta^{(0)} & = & (-2 i \delta_x - 4\Gamma_1 \delta_x r - 2 i \Gamma_1 r_x ) dx + (2\Gamma_0 r -2 \delta_t i) dt
\; , \\
\Theta^{(1)} & = & -(4 \delta_x r + 2 i r_x )\Gamma_2 dx + (i + 2 r \Gamma_1 )dt \; ,  \\
\Theta^{(n)} & = & -(4 \delta_x r + 2 i r_x )\Gamma_{n+1} dx + 2 r \Gamma_n dt  \; , \qquad n \geq 2\; . 
\end{eqnarray*}
Our goal is to prove that {\em the subsequence}
\[
Im(\Theta^{(0)}),\ Im(\Theta^{(2)}),\ Im(\Theta^{(4)}),\dots
\]
{\em consists of non-trivial conservation laws that are pairwise different}. The proof is not obvious. First of all, we need to work ``on shell", on appropriate open subsets
of the full equation manifold $S^\infty$ of \eqref{eq:1}--\eqref{eq:3}. This system implies the equations 
\begin{equation}
q_t=-2rr_x,
\qquad
r_{xt}=2qr+4r\delta_x\delta_t,
\qquad
\delta_{xt}=-\frac{r_t\delta_x+r_x\delta_t}{r}\; ,
\label{eq:solved-system}
\end{equation}
and therefore a natural system of local coordinates on $\{ r \neq 0 \} \subset S^\infty$ is given by 
\begin{equation} \label{cs}
x,\ t,\qquad q_{x^k}, \ k\ge 0; \qquad r_{x^k}, \ k\ge 0; \qquad r_{t^\ell}, \ \ell\ge 1;\qquad  \delta_{x^k}, \ k\ge 0; \qquad  \delta_{t^\ell}, \ \ell\ge 1\; .
\end{equation}
Note that we can see by induction using the explicit formula for $\Gamma_0$ and Equation (\ref{gs}), that the density $\varrho_{n}$ of the conservation law 
$Im(\Theta^{(n)})=\varrho_{n}\,dx+J_{n}\,dt$, $n \geq 0$, depends only on $x$-derivatives on $S^\infty$.  Thus, of particular interest for us is the subalgebra of functions on $S^\infty$ depending only on
\[
x,\ q_{x^k},\ r_{x^k},\ \delta_{x^k},\qquad k\ge 0 \; .
\]
On this subalgebra, the total derivative operator $D_x$ becomes
\begin{equation}  \label{dx}
D_x
=
\frac{\partial}{\partial x}
+\sum_{k\ge 0} q_{x^{k+1}}\frac{\partial}{\partial q_{x^k}}
+\sum_{k\ge 0} r_{x^{k+1}}\frac{\partial}{\partial r_{x^k}}
+\sum_{k\ge 0} \delta_{x^{k+1}}\frac{\partial}{\partial \delta_{x^k}}\; ,
\end{equation}
which is the standard one-dimensional total derivative in the $x$-direction. The following result, our main theorem on conservation laws, is motivated by Proposition \ref{spin}.

\begin{theorem}
We denote by $S^\infty_\kappa \subset S^\infty$ the submanifold of $S^\infty$ determined by the generalized Konno--Oono system \eqref{eq:1}--\eqref{eq:3} together with
the constraint
\begin{equation} \label{co}
q^2+r_x^2+4r^2\delta_x^2=\kappa(x)\; .
\end{equation}
We work on an open subset $\mathcal U_\kappa \subset S^\infty_\kappa$ in which
\[
r\neq0,\qquad q\neq0,\qquad Q=q+s\sqrt{\kappa(x)} \neq 0\; ,
\]
with $s=\pm1$. The pull-backs of the one-forms $Im(\Theta^{(2m)})$, $m=1,2,\cdots$, to $\mathcal U_\kappa \subset S^\infty_\kappa$ define non-trivial, linearly independent horizontal cohomology classes.
In particular, they are pairwise distinct non-trivial conservation laws of the augmented system  \eqref{eq:1},  \eqref{eq:2}, \eqref{eq:3},  \eqref{co}.
\end{theorem}
\begin{proof}

The proof is long. We have divided it in several steps so that its main structure is clear. 

\paragraph{Step I. A geometrically motivated change of coordinates in $\mathcal U_\kappa \subset S^\infty_\kappa$.}

We begin by noticing that on $S^\infty_\kappa$, the formula for $\Gamma_0$ and Equation (\ref{gs}) become
\begin{equation} \label{g00}
\Gamma_0 = \frac{-q + s \sqrt{\kappa(x)}}{\kappa(x)-q^2} (r_x + 2 i \delta_x r)  = \frac{1}{Q} (r_x + 2 i \delta_x r) \; , 
\end{equation}
and
\begin{equation} \label{rf}
\Gamma_{n+1}
= \frac{1}{2 s\,i\sqrt{\kappa(x)}}\,\left( -\Gamma_{n,x}-2i\delta_x\Gamma_n - (2r\delta_x+ir_x)\sum_{j=1}^{n}\Gamma_j\Gamma_{n+1-j} \right) 
\end{equation}
respectively. Since $q\neq0$, the constraint (\ref{co}) implies $\kappa(x)>0$ on $\mathcal U_\kappa$. Equation (\ref{co}) also implies that $q$ and its $x$-derivatives are no longer 
independent variables on $S^\infty_\kappa$. This fact makes direct computations with $\Gamma_n$ and $\Theta^n$ difficult: on the one hand, the overall coefficient $1/(2 s i \sqrt{S})$ appearing in
(\ref{gs}) becomes very simple, as (\ref{rf}) shows; on the other hand, the coefficient $Q$ appearing in (\ref{g00}) is no longer a zero-order function on $S^\infty_\kappa$, and this makes the recursive formula
(\ref{rf}) very hard to handle. It no longer separates the highest $x$-derivative jets in a transparent way.

Our solution is to make explicit a geometric fact which is hidden in formula (\ref{g00}). We set $K(x):=\sqrt{\kappa(x)}$  and we consider the spin vector
\[
X=(X_1,X_2,X_3):=(q,r_x,2r\delta_x) 
\]
that we introduced in Proposition \ref{spin}. As we showed therein,  we have
\[
X_1^2+X_2^2+X_3^2=K(x)^2\; 
\]
on $S^\infty_\kappa$, that is, for each fixed $x$, the vector $X$ lies on the sphere of radius $K(x)$. 
{\em We define $\zeta$ as the complex stereographic coordinate on this sphere  projected from the pole $(-sK,0,0)$ onto the 
equatorial plane $X_1=0$}, namely, 
\[
\zeta = \alpha + i \beta = sK \frac{r_x+2ir\delta_x}{q+sK} = sK \frac{r_x}{q+sK}  + i sK \frac{2r\delta_x}{q+sK} \; .
\]
Then, (\ref{g00}) becomes
$$\Gamma_0 = \frac{\zeta}{sK}\; .$$
This means that {\em we can consider $\Gamma_0$ itself as essentially a new coordinate on the open subset $\mathcal U_\kappa\,$}. The inverse parametrization is
\begin{equation} \label{nc}
q=sK\,\frac{K^2-|\zeta|^2}{K^2+|\zeta|^2}\; , \qquad r_x=\frac{2K^2\,\alpha}{K^2+|\zeta|^2}\; , \qquad 2r\delta_x=\frac{2K^2\,\beta}{K^2+|\zeta|^2}\; ,
\end{equation}
and in particular, we find
\[
Q=q+sK=\frac{2sK^3}{K^2+|\zeta|^2}\; .
\]

We use this stereographic coordinate as part of a new local coordinate chart on $\mathcal U_\kappa$. More precisely, comparing with (\ref{cs}), 
now we have coordinates  $x,t,r,\delta, \alpha, \beta$ and higher  $x$-derivatives of $\alpha, \beta$; we do not need $x$-derivatives of $r$ and $\delta$ as coordinates since they are computed via 
\[
D_x r=\frac{2K^2\alpha}{K^2+\alpha^2+\beta^2}\; ,
\qquad \mbox{ and } \qquad 
D_x\delta=\frac{K^2\beta}{r(K^2+\alpha^2+\beta^2)}\; ,
\]
and neither do we use variables $q,q_x,q_{xx},\ldots$, as we already pointed out. Also, we keep the $t$-derivatives of $r$ and $\delta$, but we do not
consider as coordinates the $t$-derivatives of $\alpha$ and $\beta$, since
straightforward computations imply that 
\begin{equation} \label{ncab}
\alpha_t = \frac{rs}{K}(K^2 + \alpha^2 - \beta^2) + 2 \beta \delta_t \; , \qquad \beta_t = \frac{2s r }{K} \alpha \beta - 2 \alpha \delta_t\; .
\end{equation}
Thus,  a complete set of ``stereographic coordinates'' on $\mathcal U_\kappa$ is
\begin{equation} \label{ncs}
x,\ t,\ r,\ \delta,\ 
\alpha,\beta,\alpha_x,\beta_x,\alpha_{xx},\beta_{xx},\ldots,
\delta_t,\delta_{tt},\delta_{ttt},\ldots,
r_t,r_{tt},r_{ttt},\ldots\;  .
\end{equation}

\paragraph{Step II. Conserved densities in stereographic coordinates.}

In the new coordinates (\ref{ncs}), the conserved density of $Im(\Theta^{(2m)})$ admits a simple expression. 
Let us write $\Gamma_n=A_n+iB_n$. Using  (\ref{rf}), it is easy to see that 
\begin{equation} \label{AB_full_rec}
A_{n+1}
= \frac{1}{2s\sqrt{\kappa(x)}} \Bigl( -B_{n,x} -2\delta_x A_n -(2r\delta_xQ_n+r_xP_n) \Bigr),
\end{equation}
\begin{equation}\label{AB_full_rec_im}
B_{n+1}
= \frac{1}{2s\sqrt{\kappa(x)}} \Bigl( A_{n,x} -2\delta_x B_n +(2r\delta_xP_n-r_xQ_n) \Bigr) \; ,
\end{equation}
in which
\[
P_n:=\sum_{j=1}^n\bigl(A_jA_{n+1-j}-B_jB_{n+1-j}\bigr),
\qquad
Q_n:=\sum_{j=1}^n\bigl(A_jB_{n+1-j}+B_jA_{n+1-j}\bigr)\; .
\]
Since  $\rho_n = -(4r\delta_x+2ir_x)\Gamma_{n+1}\,$, we obtain  $\varrho_n = {Im}(\rho_n) = -4r\delta_x\,B_{n+1}-2r_x\,A_{n+1}\,$, 
and therefore $\varrho_{2m} = -4r\delta_x\,B_{2m+1} -2r_x\,A_{2m+1}$.  Using (\ref{nc}) we obtain
\[
\varrho_{2m} = -\frac{4K^2}{K^2+|\zeta|^2} \left( \alpha A_{2m+1}+\beta B_{2m+1} \right)\; ,
\]
that is, 
\begin{equation} \label{cdc}
\varrho_{2m} = -\frac{4K^2}{K^2+|\zeta|^2} Re\bigl(\overline{\zeta}\,\Gamma_{2m+1}\bigr)\; .
\end{equation}

We claim that, for every $n\geq0$,
\begin{equation} \label{(46)}
\Gamma_n = \frac{i^n}{(2sK)^n sK}\,\zeta_{x^n} + \Gamma^{< n}\; ,
\end{equation}
where $\Gamma^{< n}$ involves only $\zeta,\zeta_x,\ldots,\zeta_{x^{n-1}}$ 
and their complex conjugates, in which here and henceforth, $\zeta_{x^j}$ for instance, means $\zeta_{xx\cdots x}$, the variable $x$ appearing $j$ times. We proceed by induction. For $n=0$, (\ref{(46)}) becomes, 
precisely $\Gamma_0=\zeta/sK$.  Assume that (\ref{(46)}) holds up to $n=k$. Equation (\ref{rf}) implies that the only term which can contain $\zeta_{x^{k+1}}$ is $-\Gamma_{k,x}$, since 
$\delta_x$ and $r_x$ are zero-order functions of $\alpha=(1/2)(\zeta+\overline{\zeta})$ and $\beta=(1/2i)(\zeta - \overline{\zeta})$, and the
remaining terms involve only $\Gamma_j$ with $j \leq k$. Hence,  using (\ref{rf}) we find 
\[
\Gamma_{k+1}
=
\frac{i^{k+1}}{(2sK)^{k+1}sK}\,\zeta_{x^{k+1}}
+
\Gamma^{< k+1}\; ,
\]
in which $\Gamma^{< k+1}$ involves only $x$-derivatives of $\zeta$ and $\overline{\zeta}$ of order at most $k$. Taking $n=2m+1$, we obtain
\[
\Gamma_{2m+1}
=
\frac{(-1)^mi}{2^{2m+1}K^{2m+2}}\,\zeta_{x^{2m+1}} + \Gamma^{< 2m+1}\; ,
\]
and substituting this expression into (\ref{cdc}), we find an expression for $\varrho_{2m}\,$: 
\begin{equation} \label{rho2m}
\varrho_{2m} = c_m\,(\alpha \beta_{x^{2m+1}} - \beta \alpha_{x^{2m+1}} ) + P_m\; ,
\end{equation}
with
\[
c_m = \frac{(-1)^m}{2^{2m-1}K^{2m}(K^2+|\zeta|^2)} \neq 0,
\]
and where $P_m$ contains no terms involving $\alpha_{x^{2m+1}}$ or $\beta_{x^{2m+1}}$. 

\smallskip

For later use, we set
\[
\Omega_{p,q} =\alpha_{x^p} \beta_{x^q} - \beta_{x^p} \alpha_{x^q}\; ,
\]
in which $\alpha_{x^0} = \alpha$ and analogously for $\beta$, so that $\varrho_{2m} = c_m \Omega_{0,2m+1} + P_m$. 

\paragraph{Step III. Analysis of the remainder term $P_m$.}

We need to study the term $P_m$ appearing in (\ref{rho2m}). Since (\ref{cdc}) and (\ref{(46)}) are written in terms of $\zeta$ and $\overline{\zeta}$, it is natural to use these
complex variables instead of $\alpha$ and $\beta$: a monomial depending on $x$-derivatives will be written in the form
\begin{equation} \label{monomial}
M=f(x,t,r,\delta,\zeta,\overline{\zeta})   \prod_{a=1}^{u}\zeta_{x^{p_a}}  \prod_{b=1}^{v}\overline{\zeta}_{x^{q_b}}\; , \qquad p_a,q_b\geq1\; ,
\end{equation}
in which an empty product ({\em i.e.}, $u$ or $v$ equal to 0) is understood to be equal to $1$. We define {\em the number of derivative factors} of $M$ and its {\em total order} by
\[
\nu(M)=u+v,
\qquad
o(M)=\sum_{a=1}^{u}p_a+\sum_{b=1}^{v}q_b
\]
respectively. We make three claims on the structure of $\Gamma_n$:
\begin{enumerate}
\item Every monomial occurring in $\Gamma_n$, $n \geq 0$, has total order at most $n$. Thus, (\ref{cdc}) implies that $\varrho_{2m}$ has total order at most $2m+1$.
\item The part of $\Gamma_n$, $n \geq 1$, containing exactly one $x$-derivative term and having total order $n$ is 
$$
\lambda_n(x) \zeta_{x^n}\; , \qquad  \mbox{ where } \qquad \lambda_n(x) =  \frac{i^n}{(2sK)^n sK} \; .
$$
Hence, $\varrho_{2m}$ has total order exactly equal to $2m+1$, and $P_m$ has total order at most $2m+1$.
\item The part of $\Gamma_n$, $n \geq 0$, containing exactly two $x$-derivative terms and having total order $n$ has the form
\begin{equation} \label{alfa}
\sum_{\substack{p+q=n\\p,q\geq1}} h_{n;p,q}(x,t,r,\delta,\zeta,\overline{\zeta}) \,\zeta_{x^p}\zeta_{x^q}.
\end{equation}
In particular, it contains no mixed monomial $\zeta_{x^p} \overline{\zeta}_{x^q}\,$.
\end{enumerate}

\noindent
If $n=0$ all three assertions are immediate, since $\Gamma_0= \zeta/sK$.
We prove the first claim by induction. The $n=0$ part is already done, let us assume that it holds up to $n$. We consider $\Gamma_{n+1}$ as given by (\ref{rf}), 
and we note that (\ref{nc}) implies that 
\[
\delta_x =\frac{K^2(\zeta - \overline{\zeta})}{2 i r(K^2+|\zeta|^2)}\; , \qquad 2r\delta_x+ir_x = \frac{2iK^2\overline{\zeta}}  {K^2+|\zeta|^2}
\]
are zero-order functions in stereographic coordinates. We have: $D_x$ increases total order of a jet monomial by at most one; 
multiplication by $\delta_x$ or $2r\delta_x+ir_x$ does not increase total order; and the total order of the terms $\Gamma_j \Gamma_{n+1-j}$ is at most $n+1$. It follows that every monomial in 
$\Gamma_{n+1}$ has total order at most $n+1$. 

The second claim follows from (\ref{(46)}). If $\Gamma^{<n}$ contains a term of total order $n$, it has to be a term with two or more factors, since it involves only terms
$\zeta_{x^j}$ and/or $\overline{\zeta}_{x^j}$ up to $j=n-1$. Hence, we are left precisely with $\lambda_n \zeta_{x^n}$. 

We prove the third assertion by induction: we assume it for $\Gamma_{n}$ and we consider all terms appearing in $\Gamma_{n+1}$.  In fact, looking at (\ref{rf}), we see that it is enough to consider the terms 
appearing in  $2 s i K\, \Gamma_{n+1}$, since the coefficient $2 s i K$ is of $x$-order zero.

\noindent We first analyze the term $D_x\Gamma_n$.  Let $M$ be a monomial as in (\ref{monomial}) appearing in $\Gamma_n$.
Its total $x$-derivative is, in full detail, 
\begin{equation} \label{der}
\begin{aligned}
D_xM
={}&
\left(
f_x+f_r r_x+f_\delta\delta_x
\right)
\prod_{a=1}^{u}\zeta_{x^{p_a}}
\prod_{b=1}^{v}\overline{\zeta}_{x^{q_b}} 
+
\left(
f_\zeta\zeta_x +
f_{\overline{\zeta}}\overline{\zeta}_x
\right)
\prod_{a=1}^{u}\zeta_{x^{p_a}}
\prod_{b=1}^{v}\overline{\zeta}_{x^{q_b}}
\\
&+
f\sum_{a=1}^{u}
\zeta_{x^{p_a+1}}
\prod_{\substack{c=1\\ c\neq a}}^{u}\zeta_{x^{p_c}}
\prod_{b=1}^{v}\overline{\zeta}_{x^{q_b}}
+
f\sum_{b=1}^{v}
\overline{\zeta}_{x^{q_b+1}}
\prod_{a=1}^{u}\zeta_{x^{p_a}}
\prod_{\substack{c=1\\ c\neq b}}^{v}
\overline{\zeta}_{x^{q_c}}\; .
\end{aligned}
\end{equation}
\begin{itemize}
\item[$\bullet$] If $u=v=0$, then only the second summand of (\ref{der}) has $x$-derivatives, and this summand contains only one-factor terms.
Thus, a zero-order term of $\Gamma_n$ cannot produce a two-factor term in $D_x\Gamma_n$.
\item[$\bullet$] If $M$ contains exactly one $x$-derivative factor, a two-factor term of 
total order $n+1$ can be produced only if $M$ has total order $n$. By our second claim, the unique one-factor term of total order
$n$ in $\Gamma_n$ is $\lambda_n \zeta_{x^n}$. Since $\lambda_n$ depends only on $x$, no two-factor terms can appear. 
\item[$\bullet$] If $M$ contains exactly two $x$-derivative factors, then a two-factor term of total order $n+1$ can arise only from the last two
summands of (\ref{der}), and only when $M$ itself has total order $n$. By the induction hypothesis, the two-factor, total order $n$ component of $\Gamma_n$ is of the form 
\[
\sum_{\substack{p+q=n\\p,q\geq1}} h_{n;p,q}\zeta_{x^p} \zeta_{x^q}\; ,
\]
and (\ref{der}) implies that in this case we obtain the term 
\[
\sum_{\substack{p+q=n\\p,q\geq1}}
h_{n;p,q}
\left(
\zeta_{x^{p+1}} \zeta_{x^q} + \zeta_{x^p} \zeta_{x^{q+1}}
\right),
\]
which is precisely of the form (\ref{alfa}) asserted in the third claim.
\item[$\bullet$] Finally, if $M$ contains at least three $x$-derivative factors,
then every term in $D_xM$ contains at least three such factors.
\end{itemize}
Thus, the only contribution of $D_x\Gamma_n$ to the component of $\Gamma_{n+1}$ having exactly two $x$-derivative factors and total order $n+1$, is of the form (\ref{alfa}) and it
comes from differentiating the two-factor total order $n$ component of $\Gamma_n$.

\noindent Now we analyse the other terms of $\Gamma_{n+1}$. The term $\delta_x\Gamma_n$ gives no contribution of the required type, since $\delta_x$ is a zero-order function and, by our first
claim, every monomial in $\Gamma_n$ has total order at most $n$. It remains to consider the quadratic term
\[
(2r\delta_x+ir_x) \sum_{j=1}^{n}\Gamma_j\Gamma_{n+1-j} = \frac{2iK^2\overline{\zeta}}  {K^2+|\zeta|^2}  \sum_{j=1}^{n}\Gamma_j\Gamma_{n+1-j} \; .
\]
Choose a monomial $M_1$ in $\Gamma_j$ and a monomial $M_2$ in $\Gamma_{n+1-j}$. For their product to have two $x$-derivative factors and total order
$n+1$, we must have
\[
\nu(M_1) + \nu(M_2)=2\; , \qquad \mbox{ and } \qquad o(M_1)+o(M_2)=n+1\; .
\]
Our first claim gives $o(M_1) \leq j$ and $o(M_2) \leq n+1-j$. Hence $o(M_1)+o(M_2)=n+1$ is possible only if $o(M_1)=j$ and $o(M_2)=n+1-j$. 
This observation implies that $\nu(M_1) = \nu(M_2) = 1$.  Indeed, if $\nu(M_1)=0$, then $o(M_1)=j=0$, but the sum in (\ref{rf}) starts from $j = 1$. Similarly, $\nu(M_2)=0$ is impossible because
it would imply that $o(M_2) = 0$, but $o(M_2)=n+1-j\geq1$.  It then follows from our second claim that $M_1$ and $M_2$ must be $\lambda_j\zeta_{x^j}$ and $\lambda_{n+1-j}\zeta_{x^{n+1-j}}$
respectively, and therefore the only two-$x$-derivative-factor term of total order $n+1$ appearing in the quadratic term $\Gamma_j \Gamma_{n+1-j}$ is precisely 
\[
-  \frac{2iK^2\overline{\zeta}}  {K^2+|\zeta|^2}  \sum_{j=1}^{n} \lambda_j\lambda_{n+1-j} \zeta_{x^j}\zeta_{x^{n+1-j}}\; .
\]
Combining the contributions to $\Gamma_{n+1}$ coming from $D_x\Gamma_n$ and the quadratic term, we see that the part of $\Gamma_{n+1}$ containing exactly
two $x$-derivative factors and having total order $n+1$ is indeed of the form
\[
\sum_{\substack{p+q=n+1\\p,q\geq1}} h_{n+1;p,q} \zeta_{x^p}\zeta_{x^q}\; .
\]
This observation ends our proof of the third claim.

\medskip

Now we go back to $P_m$ using Equations (\ref{cdc}) and (\ref{rho2m}) for $\varrho_{2m}$. We have:
\begin{itemize}
\item[$\bullet$]
Since $P_m$ contains no terms involving $x$-derivatives of order $2m+1$, every term of $P_m$ containing at most one  $x$-derivative has total order at most $2m$.
\item[$\bullet$]
Since (\ref{cdc}) holds, the terms of $P_m$ containing exactly two $x$-derivatives and having total order $2m+1$ can come only from $\Gamma_{2m+1}$. By the third claim we just proved, terms 
of this kind in $\Gamma_{2m+1}$ have the form
\[
\sum_{\substack{p+q=2m+1\\ p,q\geq1}} h_{p,q}\, \zeta_{x^p}\zeta_{x^q}\; ,
\]
with coefficients $h_{p,q}$ of order zero. Using $Re(z) =(1/2)(z + \overline{z})$ and recalling (\ref{cdc}), we see that their contribution to $P_m$ is 
\[
-\frac{2K^2}{K^2+|\zeta|^2} \sum_{\substack{p+q=2m+1\\ p,q\geq1}} \left( \overline{\zeta}\,h_{p,q}\, \zeta_{x^p}\zeta_{x^q}
+
\zeta\,\overline{h_{p,q}}\, \overline{\zeta}_{x^p}\overline{\zeta}_{x^q} \right)\; .
\]
\end{itemize}

\paragraph{Step IV. Construction of an appropriate representative of $P_m$.}
Now that we have a good control of the terms appearing in $P_m\,$, we modify this remainder by adding appropriate total derivatives, that is, by changing the representative of the conservation law
$Im(\Theta^{(2m)})$. 
The procedure below alters only the remainder, and not the distinguished term $\Omega_{0,2m+1}=\alpha\,\beta_{x^{2m+1}}-\beta\,\alpha_{x^{2m+1}}$ appearing in 
$\varrho_{2m}$. 

We write
\[
F\equiv G \pmod{\operatorname{Im}D_x}
\]
whenever $F-G=D_xH$ for some differential function $H$.  We note that for every differential function $V$ and every $\ell\geq1$,
\[
D_x\bigl(u_{x^{\ell-1}}V\bigr)
=
u_{x^\ell}V+u_{x^{\ell-1}}D_xV,
\qquad
u \in \{\alpha,\beta,\zeta,\overline{\zeta}\}\; ,
\]
and consequently,
\begin{equation}\label{leibniz1}
u_{x^\ell}V
\equiv
-u_{x^{\ell-1}}D_xV
\pmod{\operatorname{Im}D_x}\; .
\end{equation}
Iterating \eqref{leibniz1}, we obtain, for
$0\leq s\leq \ell$,
\begin{equation}\label{leibniz2}
u_{x^\ell}V
\equiv
(-1)^s u_{x^{\ell-s}}D_x^sV
\pmod{\operatorname{Im}D_x}.
\end{equation}


\noindent {\bf Claim.} 
Let
\begin{equation} \label{claim1}
T=a\prod_{i=1}^{\nu}u^i_{x^{k_i}}\; ,
\qquad
u^i\in\{\alpha,\beta,\zeta,\overline{\zeta}\}\; , 
\qquad
k_i\geq1\; ,
\end{equation}
where $a$ is a zero-order coefficient and either 
\[
(a) \qquad  \sum_{i=1}^{\nu}k_i \leq 2m  \; , \qquad \qquad \mbox{ or } \qquad \qquad (b) \qquad  \nu \geq 3 \mbox{ and } \sum_{i=1}^{\nu}k_i=2m+1\; .
\]
Then $T$ is congruent modulo $\operatorname{Im}D_x$ to a finite linear combination of monomials such that no monomial contains an
$x$-derivative factor of order greater than $m$. 

We note the following elementary fact, which follows from the Leibniz rule and the observation that $D_x r$ and
$D_x\delta$ have order zero: if $a=a(x,t,r,\delta,\alpha,\beta)$ is a zero-order coefficient, every $x$-derivative monomial in $D_x^{j}a$ has total order at most $j$. 
Now, if all numbers $k_i$ are at most equal to $m$, there is nothing to prove. Otherwise, after reordering the factors of $T$, we let
\[
k_1=\max_{1\leq i\leq\nu}k_i>m \; ,
\]
set $k_1 = m+p$, $p \geq 1$, and we define
\[
M=\prod_{i=2}^{\nu}u^i_{x^{k_i}}\; ,
\]
so that $T = a u^1_{x^{k_1}} M$. By (\ref{leibniz2}) with $\ell = k_1$ and $s=p$, we find
\[
T\equiv
(-1)^{p}u^1_{x^m}D_x^{p}(aM)
\pmod{\operatorname{Im}D_x}\; ,
\]
and the generalized Leibniz rule gives
\[
D_x^{p}(aM)
=
\sum_{\substack{j_0+j_2+\cdots+j_\nu=p\\
                 j_0,j_2,\ldots,j_\nu\geq0}}
\frac{p!}{j_0!j_2!\cdots j_\nu!}
(D_x^{j_0}a)
\prod_{i=2}^{\nu}u^i_{x^{k_i+j_i}}\; . 
\]
We assume $(a)$. We observe that the total order of each monomial appearing in this sum is at most 
$$
j_0 + k_2 + \cdots + k_\nu + j_2 + \cdots + j_\nu = p + k_2 + \cdots + k_\nu = k_1 + k_2 + \cdots + k_\nu - m \leq 2m - m  = m \; ,
$$
since monomials appearing in $D_x^{j_0}a$ have order at most $j_0$. 
Now we assume $(b)$.  Since $\nu\geq3$ and every $k_i\geq1$, we must have $m < k_1 \leq 2m-1$, and hence $1\leq p\leq m-1$. 
It follows that every $x$-derivative factor of a monomial appearing in $D_x^{j_0}a$ has order at most $j_0\leq p\leq m-1$.
Also, since $j_i\leq p$, we have $k_i+j_i \leq k_i+p = k_i+k_1-m$. 
For each fixed $i\geq2$, choose an index $i_0$ different from $1$ and $i$ (such an index exists, as $ \nu \geq 3$). Since $k_{i_0} \geq1$, we have
\[
k_1+k_i \leq \sum_{j=1}^{\nu}k_j - k_{i_0} \leq (2m+1)-1 = 2m\; ,
\]
and therefore $k_i+j_i \leq k_i+k_1-m \leq m$.  Thus every $x$-derivative factor occurring in monomials appearing in $D_x^p(aM)$ has order at most $m$. The proof of the claim is now complete.

\medskip

We apply this result to the monomials occurring in $P_m$. It implies that we can modify appropriately all monomials appearing in $P_m$, except those containing exactly two positive 
$x$-derivative factors and having total order $2m+1$. According to Step III, we can codify these terms into a sum of the form 
\begin{equation}\label{eq:exceptional-two-factor-terms}
\sum_{\substack{p+q=2m+1\\p,q\geq1}} \left( A_{p,q}\zeta_{x^p}\zeta_{x^q} + \overline{A_{p,q}}\, \overline{\zeta}_{x^p}\overline{\zeta}_{x^q} \right),
\end{equation}
in which every coefficient $A_{p,q}$ is a zero-order function. Let us consider one unbarred summand and assume  $1 \leq p\leq m$ and $q\geq m+1$. 
We have $q-m=m+1-p$ and  $1\leq q-m\leq m$.  Applying \eqref{leibniz2} with $\ell=q$, $s=q-m$, and  $V=A_{p,q}\zeta_{x^p}$, 
we obtain
\[
A_{p,q}\zeta_{x^p}\zeta_{x^q}
\equiv
(-1)^{q-m}\zeta_{x^m}
D_x^{q-m}\bigl(A_{p,q}\zeta_{x^p}\bigr)
\pmod{\operatorname{Im}D_x}\; ,
\]
and therefore the generalized Leibniz rule implies that 
\begin{align}
A_{p,q}\zeta_{x^p}\zeta_{x^q}
\equiv{}&
(-1)^{q-m}
A_{p,q}\zeta_{x^m}\zeta_{x^{m+1}} \nonumber 
\\
&+
(-1)^{q-m}
\sum_{j=1}^{q-m}
\binom{q-m}{j}
(D_x^jA_{p,q})
\zeta_{x^m}\zeta_{x^{m+1-j}}
\pmod{\operatorname{Im}D_x}.  \label{cong}
\end{align}
For every $j\geq1$, the two explicit $x$-derivative factors appearing in this expressions have orders $m$ and  $m+1-j\leq m$. 
Moreover, since $A_{p,q}$ is a zero-order coefficient, every positive $x$-derivative factor occurring in $D_x^jA_{p,q}$ has order at most
$j \leq q-m\leq m$.  Thus, we only need to reduce $A_{p,q}\zeta_{x^m}\zeta_{x^{m+1}}$. 
By the Leibniz rule,
\[
D_x\left(
\frac12A_{p,q}\zeta_{x^m}^{\,2}
\right)
=
\frac12(D_xA_{p,q})\zeta_{x^m}^{\,2}
+
A_{p,q}\zeta_{x^m}\zeta_{x^{m+1}}\; ,
\]
and therefore
\[
A_{p,q}\zeta_{x^m}\zeta_{x^{m+1}}
\equiv
-\frac12(D_xA_{p,q})\zeta_{x^m}^{\,2}
\pmod{\operatorname{Im}D_x}.
\]
Since every positive $x$-derivative factor occurring in $D_xA_{p,q}$ has order one, we see that, after replacing $A_{p,q}\zeta_{x^m}\zeta_{x^{m+1}}$ by its equivalent expression above,
the right-hand side of (\ref{cong}) contains no $x$-derivative factor of order greater than $m$. The conjugate terms are treated in exactly the same way.

Summarizing and going back to the real variables $\alpha$ and $\beta$, we see that the preceding constructions allow us to  obtain differential functions $G_m$ and $P_m^{\mathrm{red}}$ such that
$P_m=D_xG_m+P_m^{\mathrm{red}}$,  in which 
\[
P_m^{\mathrm{red}}
=
P_m^{\mathrm{red}}
\bigl(
x,t,r,\delta,\alpha,\beta,
\alpha_x,\beta_x,\ldots,
\alpha_{x^m},\beta_{x^m}
\bigr).
\]

\paragraph{Step V. Construction of an appropriate representative of $\varrho_{2m}$.}

We define a new representative of the conserved density $\varrho_{2m}$ as $\varrho_{2m}^{\mathrm{new}} = \varrho_{2m}-D_xG_m\,$, and we obtain the crucial expression
\[
\varrho_{2m}^{\mathrm{new}} = c_m\Omega_{0,2m+1}+P_m^{\mathrm{red}}\; .
\]
In this step we modify the first summand of ${\varrho}_{2m}^{\mathrm{new}}$. We use the following procedure. Since $$D_x\Omega_{p,q-1} = \Omega_{p+1,q-1}+\Omega_{p,q}\; ,$$ 
we have, for a zero-order coefficient $f$, $$f\,\Omega_{p,q} = D_x(f\,\Omega_{p,q-1}) - f\,\Omega_{p+1,q-1} - (D_xf)\Omega_{p,q-1}\; ,$$
that is,
\begin{equation}\label{normalization}
f\,\Omega_{p,q} \equiv -f\,\Omega_{p+1,q-1} - (D_xf)\,\Omega_{p,q-1} \qquad 
\pmod{\operatorname{Im}D_x}\; .
\end{equation} 
Repeated application of \eqref{normalization} gives 
\begin{equation} \label{norm}
c_m\Omega_{0,2m+1} \equiv (-1)^mc_m\Omega_{m,m+1} + \sum_{\ell=0}^{m-1}
(-1)^{\ell+1}(D_xc_m)\Omega_{\ell,2m-\ell}
\pmod{\operatorname{Im}D_x}\; .
\end{equation}
Our goal (in Step VI) is to assume that an appropriately modified conserved density (let us call it ${\varrho}_{2m}'$ for now) is equal to $D_xF$ for some $F$, and to 
arrive at a contradiction.  A difficulty that we quickly find is that the right hand side of (\ref{norm}) is of $x$-order $2m$ and that the contradiction is not immediate at that level. What we could do 
is to use  the equation ${\varrho}_{2m}' = D_xF$ and determine the dependence of $F$ on higher order $x$-jets, until at some level this becomes impossible. This procedure is feasible 
but rather cumbersome. Instead, we will use (\ref{norm})  and Step IV repeatedly, until we obtain a representative of ${\varrho}_{2m}$ of $x$-order $m+1$. At this 
stage a rather elegant contradiction will become apparent.

Recalling the explicit formula for $c_m$, see (\ref{rho2m}), and the fact that $r_x$ and $\delta_x$ are of $x$-order zero in stereographic coordinates, we write   
$D_xc_m = B_m + (c_m)_\alpha\alpha_x + (c_m)_\beta\beta_x$, in which $B_m$ is of $x$-order zero. Then, 
for  $1 \leq \ell \leq m-1$, each monomial in $B_m \Omega_{\ell,2m-\ell}$ contains exactly two $x$-derivative factors and has total order $2m$, while each monomial in 
$\bigl((c_m)_\alpha\alpha_x+(c_m)_\beta\beta_x\bigr) \Omega_{\ell,2m-\ell}$  contains precisely three $x$-derivative factors and has total order $2m+1$.  
The claim appearing in Step IV implies that all these terms admit representatives in which no individual $x$-derivative factor has order greater than $m$.

\smallskip

We consider the case $\ell=0$. Since the term $-B_m\Omega_{0,2m}$ has total order $2m$, it can be reduced to order at most $m$ using case $(a)$ of the claim proved in Step IV. We only need to
reduce the product
\begin{eqnarray}
 - \left(  (c_m)_\alpha\alpha_x+(c_m)_\beta\beta_x \right) \left( \alpha\beta_{x^{2m}}-\beta\alpha_{x^{2m}} \right)  =  
-\frac12 \left( \alpha(c_m)_\alpha+\beta(c_m)_\beta \right)\Omega_{1,2m} 
\nonumber\\
  +  \beta(c_m)_\alpha\alpha_x\alpha_{x^{2m}} -\alpha(c_m)_\beta\beta_x\beta_{x^{2m}}  +\frac12
\left( \beta(c_m)_\beta-\alpha(c_m)_\alpha \right) \left( \alpha_x\beta_{x^{2m}} +\beta_x\alpha_{x^{2m}} \right). 
\label{ell-zero}
\end{eqnarray}
Each monomial appearing in one of the last three terms in \eqref{ell-zero} is of the form $f\,u_xv_{x^{2m}}$, 
where $f$ is a zero-order coefficient and $u,v\in\{\alpha,\beta\}$. Applying \eqref{leibniz2} with $\ell=2m$ and $s=m$, we obtain
\begin{eqnarray}
f\,u_xv_{x^{2m}} & \equiv & (-1)^m v_{x^m}D_x^m(fu_x) \qquad \qquad \quad \pmod{\operatorname{Im}D_x} \nonumber \\
& \equiv{} &  
(-1)^m f\,v_{x^m}u_{x^{m+1}} + (-1)^m \sum_{j=1}^{m} \binom{m}{j} (D_x^jf)v_{x^m}u_{x^{m+1-j}} \quad \pmod{\operatorname{Im}D_x} \label{eq:ell-zero-reduction} \\
&\equiv{} &
\frac{(-1)^m}{2} f \left( v_{x^m}u_{x^{m+1}} - u_{x^m}v_{x^{m+1}} \right) - \frac{(-1)^m}{2} (D_x f) u_{x^m}v_{x^{m}} \nonumber \\
 &  & + \; (-1)^m \sum_{j=1}^{m} \binom{m}{j} (D_x^jf)v_{x^m}u_{x^{m+1-j}} \qquad \qquad \quad \pmod{\operatorname{Im}D_x} \label{eq:ell-zero-reduction1} 
\end{eqnarray}
in which we have separated the first term of (\ref{eq:ell-zero-reduction}) into symmetric and antisymmetric parts and used that the symmetric part is equivalent to $-((-1)^m/2)(D_xf)u_{x^m}v_{x^m}$. 
For every $j\geq1$, the two-$x$-derivative factors in
each summand of the third term of \eqref{eq:ell-zero-reduction1} have orders $m$ and $m+1-j\leq m$ respectively and, since $f$ is a zero-order coefficient, every
$x$-derivative factor occurring in $D_x^jf$ has order at most $j\leq m$. Thus every term under the summation sign contains no $x$-derivative factor of order greater than $m$. The same is true for the second
term in \eqref{eq:ell-zero-reduction1} since $D_x f$ has order at most $1$. 
Now for the explicit terms  
coming from the last three summands in (\ref{ell-zero}): the antisymmetric part of \eqref{eq:ell-zero-reduction1}  corresponding to  $f \alpha_x\alpha_{x^{2m}}$ and $f \beta_x\beta_{x^{2m}}$ vanishes, and so
these terms admit representatives of $x$-order at most $m$;  on the other hand, the antisymmetric contributions of $f  \beta_{x^m}\alpha_{x^{m+1}}$ and $f \alpha_{x^m}\beta_{x^{m+1}}$ appearing in
(\ref{eq:ell-zero-reduction1}) are
\[
\frac{(-1)^{m+1}}{4}
\left(
\beta(c_m)_\beta-\alpha(c_m)_\alpha
\right)\Omega_{m,m+1} 
\quad \mbox{ and } \quad 
\frac{(-1)^m}{4}
\left(
\beta(c_m)_\beta-\alpha(c_m)_\alpha
\right)\Omega_{m,m+1},
\]
respectively, and therefore they cancel. It remains for us to transform the term
\[
-\frac12
\left(
\alpha(c_m)_\alpha+\beta(c_m)_\beta
\right)\Omega_{1,2m}.
\]
We apply \eqref{normalization} $(m-1)$ times, with $p=1$ and $q=2m$, and we find that this term is equivalent to 
\[
\frac{(-1)^m}{2}
\left(
\alpha(c_m)_\alpha+\beta(c_m)_\beta
\right)\Omega_{m,m+1} + 
\frac12
\sum_{j=0}^{m-2}
(-1)^j
D_x\left(
\alpha(c_m)_\alpha+\beta(c_m)_\beta
\right)
\Omega_{1+j,\,2m-j-1} 
\]
mod $\operatorname{Im}D_x$, in which we assume that the sum is empty if $m=1$. The coefficient $D_x\left( \alpha(c_m)_\alpha+\beta(c_m)_\beta \right)$ 
is the sum of a zero-order term and terms linear in 
$\alpha_x,\beta_x$. The zero-order part, multiplied by
$\Omega_{1+j,\,2m-j-1}$, contains exactly two positive
$x$-derivative factors and has total order $2m$. The part linear
in $\alpha_x,\beta_x$ multiplied by $\Omega_{1+j,\,2m-j-1}$ contains three positive $x$-derivative
factors and has total order $2m+1$. Thus the claim in Step IV applies.
We conclude that the $\ell=0$ term appearing in \eqref{norm} can be reduced to
\[
\frac{(-1)^m}{2}
\left(
\alpha(c_m)_\alpha+\beta(c_m)_\beta
\right)\Omega_{m,m+1} 
\]
up to terms of $x$-order at most $m$ which we will not write explicitly. Thus, 
\begin{align}
c_m\Omega_{0,2m+1}
\equiv{}&
(-1)^m \left[ c_m + \frac12 \left( \alpha(c_m)_\alpha+\beta(c_m)_\beta \right) \right] \Omega_{m,m+1} + R_m \pmod{\operatorname{Im}D_x},
\label{eq:normalized-leading-term}
\end{align}
in which $R_m$ contains $x$-derivative terms of $x$-order at most $m$. 
We have found that
$$
\varrho_{2m}^{\mathrm{new}} \equiv{} \widetilde{d}_m \Omega_{m,m+1} + P_m^{\mathrm{red}} + R_m
\pmod{\operatorname{Im}D_x}\; ,
$$
and therefore setting $\widetilde{R}_m= R_m +P_m^{\mathrm{red}}$, we have the new representative
\[
\widetilde\varrho_{2m} = \widetilde d_m\Omega_{m,m+1} + \widetilde R_m
\]
for $\varrho_{2m}$. A straightforward computation yields 
\begin{align}
\widetilde d_m = (-1)^m
\left[
c_m
+
\frac12
\left(
\alpha(c_m)_\alpha+\beta(c_m)_\beta
\right)
\right] = \frac{1}
{2^{2m-1}K^{2m-2}
 (K^2+\alpha^2+\beta^2)^2}
\neq0\; ,
\label{eq:explicit-dm}
\end{align}
and therefore our final representative of the conserved density $\varrho_{2m}$ is
\begin{equation}
\widetilde\varrho_{2m}
=
\frac{1}{2^{2m-1}K^{2m-2}(K^2+\alpha^2+\beta^2)^2} \left( \alpha_{x^m}\beta_{x^{m+1}} - \beta_{x^m}\alpha_{x^{m+1}} \right) +
\widetilde R_m,
\label{eq:final-normalized-density}
\end{equation}
in which the $x$-order of all $x$-derivative terms appearing in $\widetilde{R}_m$ is at most $m$.

\paragraph{Step VI. Non-triviality of $\widetilde\varrho_{2m}$.}
We argue by contradiction. Suppose that
\[
\widetilde\varrho_{2m}=D_xF 
\]
for some differential function $F$. In principle, $F$ may depend on arbitrary finitely many coordinates (\ref{ncs}). However, dependence on $t$-jets of
$r,\delta$ is irrelevant for us, since $\widetilde\varrho_{2m}$  depends only on $x$-derivatives and for every $N\geq0$, $D_xD_t^N r$ and $D_xD_t^N\delta$ 
depend only on $x,t,r,\delta,\alpha,\beta, r_t,\ldots,r_{t^N},\delta_t,\ldots,\delta_{t^N}$, as we now show.

We proceed by induction. For $N=0$, (\ref{nc}) implies that $D_xr= f^r(x,\alpha,\beta)$ and $D_x\delta= f^\delta(x,r,\alpha,\beta)$.
Assume now that the statement holds up to order $N$.  We have,
\[
D_xD_t^{N+1}r = D_t\left(D_xD_t^N r\right), \qquad D_xD_t^{N+1}\delta = D_t\left(D_xD_t^N\delta\right).
\]
By the induction hypothesis, there exist functions $\Phi_N(x,t,r,\delta,\alpha,\beta, r_t,\ldots,r_{t^N},\delta_t,\ldots,\delta_{t^N})$ and 
$\Psi_N(x,t,r,\delta,\alpha,\beta, r_t,\ldots,r_{t^N},\delta_t,\ldots,\delta_{t^N})$  such that $D_xD_t^N r = \Phi_N$ and $D_xD_t^N\delta = \Psi_N$. 
Applying $D_t$, we obtain $D_xD_t^{N+1}r = D_t\Phi_N$ and $D_xD_t^{N+1}\delta = D_t\Psi_N$, and so,
using that $D_t\alpha= A(x,t,r,\alpha,\beta,\delta_t)$ and $D_t\beta= B(x,t,r,\alpha,\beta,\delta_t)$, see (\ref{ncab}),
we find that $D_xD_t^{N+1}r = D_t\Phi_N$ and  $D_xD_t^{N+1}\delta= D_t\Psi_N$ depend only on
$x,t,r,\delta,\alpha,\beta$ and pure $t$-jets of $r,\delta$ up to $t$-order $N+1$.  This proves the induction step.

Thus, without loss of generality since pure $t$-derivatives possibly appearing in $F$ play no role in our proof below, we assume that $F$ depends only on $x,t,r,\delta,\alpha,\beta$ and a finite number of 
$x$-derivatives of $\alpha$ and $\beta$.  Since $\widetilde\varrho_{2m}$ has $x$-order $m+1$, $F$ must have $x$-order at most $m$ with respect to
$\alpha,\beta$, as if $J>m$ were the largest $x$-order appearing in $F$, the coefficients $F_{\alpha_{x^J}}$ and  $F_{\beta_{x^J}}$ of $\alpha_{x^{J+1}}$ and $\beta_{x^{J+1}}$ in $D_xF$ 
would have to vanish.  Repeating this argument removes all dependence of $F$ on $x$-jets of order greater than $m$.

Now comes our final non-triviality argument. Comparison of the highest $x$-order coefficients of $\widetilde{\varrho}_{2m}$ and $D_x F$ gives
\[
F_{\beta_{x^m}}
=
\widetilde d_m\alpha_{x^m},
\qquad
F_{\alpha_{x^m}}
=
-\widetilde d_m\beta_{x^m}\; .
\]
Because $\widetilde d_m$ is independent of
$\alpha_{x^m},\beta_{x^m}$, differentiating once again yields
\[
F_{\alpha_{x^m}\beta_{x^m}}
=
\widetilde d_m,
\qquad
F_{\beta_{x^m}\alpha_{x^m}}
=
-\widetilde d_m\; .
\]
This contradicts equality of mixed partial derivatives, since
$\widetilde d_m\neq0$.  Therefore
$\widetilde\varrho_{2m}$, and hence also $\varrho_{2m}$, is not a
total $x$-derivative.  Thus
\[
\operatorname{Im}\bigl(\Theta^{(2m)}\bigr)
\]
defines a non-trivial horizontal cohomology class for every $m\geq1$.

\paragraph{Step VII. Linear independence of cohomology classes $Im\bigl(\Theta^{(2j)}\bigr)$.}

It remains to prove linear independence.  Suppose that
\[
\sum_{j=1}^{M}
a_j\,\operatorname{Im}\bigl(\Theta^{(2j)}\bigr)\; , \qquad a_j \in \mathbb{R}\; , \qquad a_M \neq 0\; ,
\]
is horizontally exact.  Replacing each density by the representative just constructed, we obtain
\begin{equation} \label{aste}
\sum_{j=1}^{M}a_j\widetilde\varrho_{2j}=D_xF
\end{equation}
for some $F$. Reasoning as in Step VI, we may assume that $F$ has $x$-order at most $M$. For $j<M$, the terms appearing in the representative $\widetilde\varrho_{2j}$ have $x$-order at
most $j+1\leq M$.  Therefore the variables $\alpha_{x^{M+1}}$ and $\beta_{x^{M+1}}$  occur only in the last summand, and (\ref{aste}) implies that 
\[
F_{\beta_{x^M}}
=
a_M\widetilde d_M\alpha_{x^M}\; ,
\qquad
F_{\alpha_{x^M}}
=
-a_M\widetilde d_M\beta_{x^M}\; .
\]
Reasoning as in the previous step we obtain a contradiction since $a_M\widetilde d_M\neq 0$. 
Hence no non-zero finite linear combination of the classes
\[
\left[
\operatorname{Im}\bigl(\Theta^{(2m)}\bigr)
\right],
\qquad m\geq1,
\]
is horizontally exact.  The conservation laws are therefore linearly
independent, and in particular pairwise distinct, in horizontal
cohomology.  

\smallskip

We have finished the proof of Theorem 5.
\end{proof}

We have proven integrability of the augmented system  \eqref{eq:1},  \eqref{eq:2}, \eqref{eq:3},  \eqref{co}. It remains to investigate the same issue in case of the full generalized Konno--Oono system. 
Indeed we have:

\begin{corollary}
The conservation laws $Im\bigl(\Theta^{(2m)}\bigr)$, $m\geq 0$, are non-trivial and linearly independent in the horizontal cohomology of (a generic open subset $\mathcal U$ of) $S^\infty$.
\end{corollary}
\begin{proof}
First we consider the case $m \geq 1$.
Fix a generic function $\kappa(x) > 0$ and a non-empty open set $\mathcal U_\kappa \subset S^\infty_\kappa$ such that the hypotheses of Theorem 5 are satisfied, and also fix an open set $\mathcal U \subset 
S^\infty$ such that the inclusion map
\[
\iota_\kappa: \mathcal U_\kappa \hookrightarrow \mathcal U
\]
is well-defined. Since $S^\infty_\kappa$ is the submanifold of $S^\infty$ obtained by imposing
\[
q^2+r_x^2+4r^2\delta_x^2=\kappa(x)
\]
and all of its differential consequences, $\iota_\kappa$ preserves the Cartan distribution. Hence
\[
\iota_\kappa^*\circ d_H=d_H\circ \iota_\kappa^*.
\]
Suppose that for some $m\geq 1$, $Im\bigl(\Theta^{(2m)}\bigr)=d_H F$ on $\mathcal U$. Pulling back by $\iota_\kappa$, we obtain
\[
\iota_\kappa^* Im\bigl(\Theta^{(2m)}\bigr) = d_H(\iota_\kappa^*F)\; ,
\]
which contradicts Theorem 5. Therefore $Im(\Theta^{(2m)})$ is not horizontally exact on $\mathcal U$. 

\smallskip

Now we consider the case $m=0$. We recall that  $S=q^2+r_x^2+4r^2\delta_x^2$ and that $Q=q+s\sqrt{S}$. Shrinking $\mathcal U$ if necessary, we assume that 
$S>0$, $Q\neq0$, $q\neq0$, $r\neq0$.
Let us write $Im\bigl(\Theta^{(0)}\bigr) = \varrho_0\,dx+J_0\,dt$. A direct computation yields 
\begin{equation}
\varrho_0 = \frac{ 2\bigl( r_x^2\delta_x +rr_x\delta_{xx} -r\delta_xr_{xx} \bigr) }{s \sqrt{S} Q} - \frac{2q\delta_x}{s \sqrt{S}}\; .
\label{eq:rho-zero-original}
\end{equation}
We apply the Euler operator with respect to $\delta$,
\[
\mathbf E_\delta
=
\frac{\partial}{\partial\delta}
-
D_x\frac{\partial}{\partial\delta_x}
+
D_x^2\frac{\partial}{\partial\delta_{xx}}
-\cdots \; ,
\]
in which $D_x$ is given by \eqref{dx}, to $\varrho_0$ and we obtain 
\begin{align}
\mathbf E_\delta(\varrho_0)
={}&
\frac{2}{sS^{5/2}}
\Big[ S\left( qrr_{xxx} + 3qr_xr_{xx} - rr_xq_{xx} + 3q^2q_x \right) \nonumber\\
&\qquad
- 3\left( qrr_{xx} + qr_x^2 - rr_xq_x + q^3 \right) \left( qq_x + r_xr_{xx} + 4rr_x\delta_x^2 + 4r^2\delta_x\delta_{xx} \right) \Big] . \nonumber 
\end{align}
In particular, the coefficient of $r_{xxx}$ in $\mathbf E_\delta(\varrho_0)$ is $2qr/(sS^{3/2})$, which is generically non-zero.
Consequently, $\mathbf E_\delta(\varrho_0) \not \equiv 0$. Since Euler operators annihilate total $x$-derivatives,
$\varrho_0$ cannot belong to $\operatorname{Im}D_x$ and so  $Im \bigl(\Theta^{(0)}\bigr)$ defines a non-trivial horizontal cohomology class.

It remains to prove linear independence. Suppose that
\[
a_0 Im \bigl(\Theta^{(0)}\bigr)
+
\sum_{m=1}^{M}
a_m Im \bigl(\Theta^{(2m)}\bigr)
=
d_HF,
\qquad a_0, \dots, a_M \in \mathbb{R}, \qquad a_M\neq0\; .
\]
Pulling this identity back to $\mathcal U_\kappa$, we see that the conserved density of $\iota_\kappa^* Im (\Theta^{(0)})\) has $x$-order at most one in stereographic coordinates
(we recall that in these coordinates $r_x$ and $\delta_x$ are zero-order functions, and that $r_{xx}$ and $\delta_{xx}$ have $x$-order at most one)
while the highest $x$-order term in the modified ({\em i.e.} $\widetilde{\varrho}_{2m}$) density of $\iota_\kappa^* Im (\Theta^{(2M)})$ has $x$-order
$M+1$. Therefore we can repeat the highest-$x$-order-jet argument of Step VII and it gives $a_M=0$, a
contradiction. Applying the same argument repeatedly yields $a_m=0$ for $m\geq1$. The identity above then reduces to $a_0 Im \bigl(\Theta^{(0)}\bigr)=d_HF$,
and the non-triviality proved above implies $a_0=0$. Hence the complete family of cohomology classes 
\[
\left[ Im \bigl(\Theta^{(2m)}\bigr) \right], \qquad m\geq0\; ,
\]
is linearly independent.
\end{proof}

\begin{remark}
The proof of Theorem $5$ is formulated for $m\geq1$ since it relies on the existence of a specific representative for the conserved density of
$Im (\Theta^{(2m)})$ whose construction requires this restriction. Thus, the case $m=0$ must be treated separately. Now, if the conserved density
of $Im (\Theta^{(0)})$ calculated in Corollary $2$ is pulled-back to $\mathcal U_\kappa \subset S^\infty_\kappa$ we find (notation as in Step I)
\[
\iota_\kappa^*\varrho_0
=
\frac{2(\alpha\beta_x-\beta\alpha_x)}{K^2+\alpha^2+\beta^2}
+
\frac{2K^2\beta(\alpha^2+\beta^2-K^2)}{r (K^2+\alpha^2+\beta^2)^2}\; ,
\]
and the mixed-derivative argument applies.  Moreover, since this pulled-back density has $x$-order one, it does not affect the
highest-$x$-jet argument of Step VII. Consequently,  we obtain non-triviality and linear independence of the complete family
\[
\left[ \iota_\kappa^* 
Im \bigl(\Theta^{(2m)}\bigr)
\right],
\qquad m\geq 0 \; ,
\]
on $\mathcal U_\kappa \subset S^\infty_\kappa$. We omit the details.
\end{remark}

\section{Analysis in the domain of travelling waves} 

A travelling wave solution to the system \eqref{eq:1}-\eqref{eq:3} is a solution of the form:
\[
q(x, t) = Q(\xi), \quad r(x, t) = R(\xi), \quad \delta(x, t) = \Delta(\xi),
\]
where $ \xi = x - vt $ is the travelling wave coordinate and $ v > 0 $ is the constant wave velocity.

We assume that the solutions $ q(x, t) $, $ r(x, t) $, and $ \delta(x, t) $ are travelling waves moving with 
constant velocity $ v $.  Thus, substituting into the original system of equations \eqref{eq:1}-\eqref{eq:3} we obtain
\begin{align}
-v Q' + 2 R R' &= 0\; , \label{eq:1'} \\
-v R'' - 2 Q R + 4 v R (\Delta')^2 &= 0\; , \label{eq:2'} \\
- v R \Delta'' - 2 v R' \Delta' &= 0\; . \label{eq:3'}
\end{align}
Now, from equation \eqref{eq:1'}, we solve for $ Q(\xi) $:
\[
-v Q' + 2 R R' = 0 \implies (-v Q + R^2)' =0.
\]
Integrating with respect to $ \xi $ we find 
\[
Q(\xi) = \frac{R^2(\xi)-C_0}{v} \; ,
\]
where $ C_0 $ is an integration constant. {\em In this section we restrict our attention to the special subclass of travelling waves for which} $C_0=0$.  
This choice already captures a distinguished family of solutions and it allows us to present the phase plane and geometric 
analysis in a simpler form.  Substituting $Q(\xi) = R^2/v$ into equation \eqref{eq:2'}, we obtain
\[
-v R'' - \frac{2 R^3}{v} + 4 v R (\Delta')^2 = 0\; ,
\]
that is, 
\begin{equation} \label{ec-for-R}
R'' = -\frac{2 R^3}{v^2} + 4 R (\Delta')^2\; .
\end{equation}
Also, from equation \eqref{eq:3'}, we find 
\[
- v R \Delta'' - 2 v R' \Delta' = 0 \implies R \Delta'' + 2 R' \Delta' = 0\; ,
\]
and so we have
\[
\bigl(R^2\Delta'\bigr)' = 2RR'\Delta'+R^2\Delta'' = R\bigl(2R'\Delta'+R\Delta''\bigr) = 0\; .
\]
Thus, $\Delta'(\xi) = C/R^2(\xi)$ for a constant $C$. {\em Hereafter we will focus on the case $C \neq 0$}. Integrating once more we find
\[
\Delta(\xi) = C \int \frac{1}{R^2(\xi)} d\xi + \Delta_0\; ,
\]
where $ \Delta_0 $ is another constant of integration. Now substituting $ \Delta' = C/R^2 $ into (\ref{ec-for-R}) we obtain
\begin{equation} \label{ec-for-R-2}
R'' = -\frac{2 R^3}{v^2} + \frac{4 C^2}{R^3}\; .
\end{equation}
This ODE can be interpreted as the equation of motion for a particle in a potential $ U(R) $ where
\begin{equation} \label{ur}
U(R) = \frac{R^4}{2 v^2} + \frac{2 C^2}{R^2}.
\end{equation}
Thus, (\ref{ec-for-R-2}) becomes 
\[
R''(\xi) = -\frac{dU}{dR}\; .
\]
Multiplying both sides of this equation by $ R'(\xi) $ and integrating, we obtain the first integral of motion (energy conservation)
\[
\frac{1}{2} R'^{\,2} + U(R) = E\; ,
\]
where $ E $ is a constant representing the total energy and the term $(1/2) R'^{\,2}$ represents kinetic energy. 

\subsection{Analysis of the Potential $ U(R) $}

\begin{itemize}
\item
\textbf{Singularities and limits.}
As $R \to 0^+$, the second term $\tfrac{2C^2}{R^2}$ blows up to $+\infty$.  
As $R \to +\infty$, the first term $\tfrac{R^4}{2\,v^2}$ dominates and $\to +\infty$.  
Hence
\[
\lim_{R\to0^+} U(R) \;=\; +\infty,
\qquad
\lim_{R\to+\infty} U(R) \;=\; +\infty.
\]
\item
\textbf{Critical point(s).}
Set $\tfrac{dU}{dR} = 0$:
\[
\frac{2\,R^3}{v^2} \;-\; \frac{4\,C^2}{R^3} \;=\; 0
\quad\Longrightarrow\quad
R^6 \,=\, 2\,C^2\,v^2
\quad\Longrightarrow\quad
R_*  \,=\, \bigl(2\,C^2\,v^2\bigr)^{\!1/6}.
\]

Since $C\neq 0$ and $v>0$, this defines a unique positive real critical point $R=R_*$.
\item
\textbf{Nature of the critical point.}
The second derivative of $U$ is
\[
\frac{d^2U}{dR^2} \;=\; \frac{6 R^2}{v^2} \;+\; \frac{12\,C^2}{R^4},
\]
which is strictly positive for $R>0$.  Thus $R_*$ is a \emph{local (and indeed global) minimum} of $U$ on the interval $R>0$.  
Because $U(R)\to+\infty$ both as $R\to0^+$ and as $R\to+\infty$, the potential has the shape of a smooth single‐well ``bowl'' on $(0,\infty)$ with a unique global minimum at $R=R_*$. 
\end{itemize}

\subsection{Phase Plane Analysis}

We now establish that for each energy level $E > U(R^\ast)$, there is exactly one nontrivial closed {\em periodic} orbit in the $(R,R')$ phase plane. Equivalently, there is a smooth periodic solution $R(\xi)$, unique up to translation of the phase variable $\xi$ (and reversal of orientation along the same closed orbit), oscillating between two turning points.
Our argument is a standard but precise application of one-dimensional conservative mechanics, see for instance \cite[Section 14]{Arnold}.

We denote the \emph{turning points} $R_{\pm}(E)$ as those where $R'(\xi)=0$ and thus $U\bigl(R_{\pm}(E)\bigr)=E$.  We prove below that:
\begin{enumerate}
\item[\textbf{(1)}]
For each $E>U(R_*)$, there are exactly two positive turning points $0<R_{-}(E)<R_{+}(E)<\infty$.
\item[\textbf{(2)}]
The orbit $\mathcal{O}_E=\bigl\{\,(R,R')\in (0,\infty)\times \mathbb{R}:\,\tfrac12(R')^2 + U(R)=E\bigr\}$ is a simple \emph{closed} curve in the $(R,R')$ phase plane.
\item[\textbf{(3)}]
As the solution $R(\xi)$ traverses the orbit $\mathcal{O}_E$, it completes a \emph{periodic} oscillation in finite  $\xi$. 
\item[\textbf{(4)}]
The period is explicitly given by a convergent integral formula and is finite for every $E>U(R_*)$.  
\end{enumerate}

First, we discuss the positivity and finiteness of the turning points. Since 
\[
U(R)\;\xrightarrow{\,R\to 0^+\,}\;+\infty 
\quad\text{and}\quad
U(R)\;\xrightarrow{\,R\to +\infty\,}\;+\infty,
\]
the equation $U(R)=E$ for a fixed $E>U(R_*)$ admits two distinct solutions in $(0,\infty)$ (one smaller than $R_*$, one larger).  We label these points $R_{-}(E)$ and $R_{+}(E)$.  By energy conservation, when $R(\xi)$ decreases to $R_{-}(E)$ or increases to $R_{+}(E)$, we have $R'(\xi)=0$.  These points cannot merge if $E>U(R_*)$, since $U''(R)>0$ ensures $U(R)$ is strictly convex and thus crosses the horizontal line $y=E$ in \emph{exactly} the two points $R_{-}(E)$ and $R_{+}(E)$ that satisfy, therefore, $0 < R_{-}(E) < R_{+}(E) < \infty$ if $E>U(R_*)$.

Now, we show that the phase-plane trajectory is a simple closed loop. For this purpose, it suffices to write 
\[
\bigl(R'(\xi)\bigr)^2 
\;=\;
2\,\bigl(E-U(R(\xi))\bigr).
\]
Between $R_{-}(E)$ and $R_{+}(E)$, one has $U(R(\xi))\le E$ so the above is nonnegative.  As $R(\xi)$ moves from $R_{-}(E)$ to $R_{+}(E)$, $R'(\xi)$ is positive (strictly, except at turning points), so $R(\xi)$ increases.  Then from $R_{+}(E)$ back to $R_{-}(E)$, one has $R'(\xi)$ negative.  Continuity in $\xi$ closes the loop in the $(R,R')$ plane: once $R$ returns to $R_{-}(E)$ with $R'<0$, it must cross that point again with $R'>0$, completing a cycle.  Because $U(R)$ is single--valued and strictly convex, no self--intersections occur; the curve is simple and closed.

Thus, it is possible to assert the existence of a finite period. Indeed, a classical result in one--dimensional conservative mechanics (often called the \emph{Structure Theorem for One--Dimensional Hamiltonian Systems}, see Arnold, \emph{op.\ cit.}, Theorem~14.1) states that for each such closed energy curve, the corresponding solution $R(\xi)$ is \emph{periodic} in $\xi$.  Here is a direct ``integral" proof:

From the energy relation, we first solve for $R'(\xi)$:
\[
R'(\xi) 
\;  = \pm\;
\sqrt{2\,\bigl(E - U(R)\bigr)}
\]
depending on whether $R(\xi)$ is increasing or decreasing.

We consider an interval of motion from $R_{-}(E)$ (where $R'=0$) to $R_{+}(E)$ (where $R'=0$ again).  Without loss of generality, we take the branch $R'(\xi)\ge 0$.  Rearrange:
\[
\frac{dR}{\sqrt{2\,\bigl(E - U(R)\bigr)}}
\;=\;
d\xi.
\]
Hence, the time for $R(\xi)$ to move from $R_{-}(E)$ to $R_{+}(E)$ is
\begin{equation}
\label{eq:HalfPeriodIntegral}
T_1(E)
\;=\;
\int_{R_{-}(E)}^{R_{+}(E)}
\frac{dR}{\sqrt{2\,\bigl(E - U(R)\bigr)}}.
\end{equation}
Because $U(R)$ is smooth on the interval $[R_{-}(E),R_{+}(E)]$, and $E - U(R)>0$ in $(R_{-}(E),R_{+}(E))$, the integrand is finite except possibly at the turning points.  But near $R_{-}(E)$, we have $E - U(R)\sim (R-R_{-}(E))U'(R_{-}(E))$, and $U'(R_{-}(E))\neq 0$ (strict convexity), so the local singularity is integrable.  Thus $T_1(E)$ converges.

After reaching $R_{+}(E)$, the motion reverses direction ($R'(\xi)<0$) and returns to $R_{-}(E)$.  Another integral of the same form yields the return--time
\[
T_2(E)
\;=\;
\int_{R_{+}(E)}^{R_{-}(E)}
\frac{dR}{-\sqrt{2\,\bigl(E - U(R)\bigr)}}
\;=\;
\int_{R_{-}(E)}^{R_{+}(E)}
\frac{dR}{\sqrt{2\,\bigl(E - U(R)\bigr)}}
\;=\;
T_1(E).
\]
Hence the full \emph{period} of one complete oscillation is
\begin{equation}
\label{eq:FullPeriod}
T(E)
\;=\;
T_1(E) + T_2(E)
\;=\;
2 \int_{R_{-}(E)}^{R_{+}(E)}
\frac{dR}{\sqrt{2\,\bigl(E - U(R)\bigr)}},
\end{equation}
which is finite and positive for each $E>U(R_*)$.

\subsection{Limits of the Period and Behavior as $E\to U(R_*)$ or $E\to +\infty$}

If $E$ approaches $U(R_*)$ from above, then $R_{-}(E)\to R_* \to R_{+}(E)$, so the turning points both coalesce to $R_*$.  We may expand $U(R)$ in a Taylor series about $R_*$:
\[
U(R)
\;=\;
U(R_*)
\;+\;
\tfrac12\,U''(R_*)\,\bigl(R-R_*\bigr)^2 
\;+\;
\mathcal{O}\bigl((R-R_*)^3\bigr).
\]
By standard perturbation of the period integral \eqref{eq:FullPeriod}, we can show (following \cite[Section 14]{Arnold} once again) that 
\[
T(E)
\;\to\;
\frac{2\pi}{\sqrt{U''(R_*)}}
\quad\text{as }\,E\downarrow U(R_*).
\]
This is precisely the \emph{linearization frequency} for small oscillations around the equilibrium $R=R_*$.  

\smallskip 

Now we observe that the function $R(\xi)$ remains strictly between $0$ and $+\infty$ for each fixed $E$:  the turning points $R_{-}(E)$ and $R_{+}(E)$ remain finite (away from $0$ and $\infty$), the motion is confined to a 
compact interval $[\,R_{-}(E),\,R_{+}(E)\,] \subset (0,\infty)$, and the integral \eqref{eq:FullPeriod} remains finite.  In fact, $R_{-}(E)\to 0^+$ or $R_{+}(E)\to +\infty$ could never happen for a fixed, finite $E$, unless $U$ had a finite limit as $R_{-}(E)\to 0^+$ or $R_{+}(E)\to +\infty$; but $U(R)\to+\infty$ as $R\to 0^+$ or $R\to +\infty$, so those endpoints are effectively prohibited. 

However, this confinement is not uniform as $E\to +\infty$. Even more, we can check that for the potential $U(R)$ given by (\ref{ur}), the equation $U(R)=E$ implies that the lower turning 
point satisfies $R_-(E)\to 0^+$, while the upper turning point satisfies $R_+(E)\to +\infty$, as $E\to +\infty$:

\smallskip

\noindent  
If $K \subset(0,\infty)$ is any fixed compact set, then continuity of $U$ implies
that $U$ is bounded on $K$, that is, there exists $C_K>0$ such that $U(R)\le C_K$ for all $R\in K$. Therefore, if 
\[
[R_-(E),R_+(E)]\subset K,
\]
then in particular $E=U(R_-(E))\le C_K$ and $E=U(R_+(E))\le C_K$,  which is impossible for arbitrarily large $E$.

\noindent
Now we prove that $R_-(E)\to0^+$. Suppose that this is false. Then there exist
$\varepsilon_0>0$ and a sequence $E_n\to+\infty$ such that $R_-(E_n)\ge \varepsilon_0$ for all $n$.
Since $R_-(E_n)<R_*$ for all $n$, it follows that $R_-(E_n)\in[\varepsilon_0,R_*]$ for all $n$. Now, the interval $[\varepsilon_0,R_*]$ is compact and $U$ is continuous, 
and therefore there exists $M_0>0$ such that $U(R)\le M_0$ for all $R\in[\varepsilon_0,R_*]$. 
Therefore,
\[
E_n=U(R_-(E_n))\le M_0
\qquad\text{for all }n,
\]
which contradicts $E_n\to+\infty$. Thus, $R_-(E)\to0^+$.  The proof that $R_+(E)\to+\infty$ is analogous. 

\smallskip

This observation allows us to investigate the asymptotic behaviour of the period $T(E)$:

\begin{proposition}
We consider $U(R)$ given by $(\ref{ur})$, and $T(E)$ given by $(\ref{eq:FullPeriod})$. We assume that 
$R_-(E)<R_+(E)$ are the turning points determined by $U(R_\pm(E))=E$, and that $C \neq 0$, $v > 0$.
Then
\[
T(E)\longrightarrow 0
\qquad
\text{ as } \qquad E\to+\infty.
\]
\end{proposition}
\begin{proof}
We set $\varepsilon_E=2C^2/E^{3/2}$ and we apply the scaling $R=E^{1/4}s$. Then, 
\[
U(R) = U(E^{1/4}s)
=
\frac{E s^4}{2v^2} + \frac{2C^2}{E^{1/2}s^2}\; ,
\]
and therefore 
\[
E-U(E^{1/4}s) = E\left(1-\frac{s^4}{2v^2}-\frac{\varepsilon_E}{s^2}\right)\; .
\]
Thus, 
\[
T(E) = 2\int_{R_-(E)}^{R_+(E)} \frac{dR}{\sqrt{2(E-U(R))}} = \sqrt{2}\,E^{-1/4} \int_{s_-(E)}^{s_+(E)} \frac{ds}{ \sqrt{1-\dfrac{s^4}{2v^2}-\dfrac{\varepsilon_E}{s^2}} } \; ,
\]
in which the new integration limits are
\[
s_\pm(E)=\frac{R_\pm(E)}{E^{1/4}}\; .
\]
Thus, it is enough to prove that the integral
\[
I_E= \int_{s_-(E)}^{s_+(E)} \frac{ds}{ \sqrt{1-\dfrac{s^4}{2v^2}-\dfrac{\varepsilon_E}{s^2}} }
\]
remains bounded as $E\to+\infty$. We let $x=s^2$ and we consider the equation
\begin{equation} \label{sx}
1-\frac{x^2}{2v^2}-\frac{\varepsilon_E}{x}=0\; ,
\end{equation}
that is, $x^3- 2v^2 x+ 2v^2 \varepsilon_E=0$. 
An elementary analysis implies that the function $f(x) =x^3- 2 v^2 x+ 2 v^2 \varepsilon_E$ has exactly two critical points, $x=\pm\sqrt{2/3} v $, and that 
 for all sufficiently large $E$, the equation $f(x)=0$ has exactly 
one negative root and two positive roots. We denote the positive roots by $0<x_- < x_+ $, and the negative one by $-\alpha_E$, $\alpha_E>0$. Hence, 
\[
2v^2 x-x^{3}- 2v^2 \varepsilon_E = (x-x_- )(x_+-x)(x+\alpha_E).
\]
Returning to the variable $s$, we obtain 
\[
1-\frac{s^4}{2v^2}-\frac{\varepsilon_E}{s^2} = \frac{(s^2-x_-)(x_+-s^2)(s^2+\alpha_E)}{2v^2 \,s^2} 
\]
and therefore
\[
I_E = \int_{s_-(E)}^{s_+(E)} \frac{\sqrt 2 v \,s\,ds} {\sqrt{(s^{2}-x_-)(x_+-s^{2})(s^{2}+\alpha_E)}}\; .
\]
Now, the endpoints $s_\pm(E)$ solve the equation 
\[
1-\frac{s^{4}}{2v^2}-\frac{\varepsilon_E}{s^{2}}=0\; ,
\]
and therefore their images $x=s_\pm(E)^2$ satisfy (\ref{sx}). Thus, they must be either $x_-$, or $x_+$. 
Since $s_-(E)<s_+(E)$ and the map $s\mapsto s^2$ is strictly increasing on $(0,\infty)$, we must have that
\[
s_-(E)^2=x_-(E),\qquad s_+(E)^2=x_+(E)\; .
\]
Thus, we can write the integral $I_E$ as  
\[
I_E = \frac{\sqrt 2 v}{2} \int_{x_-}^{x_+} \frac{dx} {\sqrt{(x-x_- )(x_+-x)(x+\alpha_E)}} \; ,
\]
and we find the bound 
\[
I_E
\le
\frac{\sqrt 2 v}{2\sqrt{x_-+\alpha_E}} \int_{x_-}^{x_+}\frac{dx}{\sqrt{(x-x_-)(x_+-x)}} \le \frac{\pi\sqrt 2\,v}{2\sqrt{x_-+\alpha_E}}\; .
\]
Now we use Viete's relations for the cubic (\ref{sx}). We obtain 
\begin{align}
x_-+x_+-\alpha_E&=0\; , \qquad  
x_-x_+-\alpha_E(x_-+x_+) \; = \; -2v^2\; , \qquad 
\alpha_E x_-x_+ \; =\; 2v^2\varepsilon_E\; . \label{eq:viete-two}
\end{align}
Equations \eqref{eq:viete-two} imply $\alpha_E=x_-+x_+$ and $x_-^2+x_-x_++x_+^2=2v^2$. Since $0<x_- <x_+$, it follows that
$2v^2 < 3x_+^2$  and hence
\[
x_+(E)>\sqrt{\frac{2}{3}}\,v.
\]
Consequently,
\[
x_-(E)+\alpha_E
=
2x_-(E)+x_+(E)
>
\sqrt{\frac{2}{3}}\,v\; ,
\]
and so we have 
\[
0<T(E)=\sqrt{2}\,E^{-1/4}I_E
\le
\sqrt{2}\,E^{-1/4}\frac{\pi\sqrt 2 v}{2\sqrt{x_-+\alpha_E}} = \frac{\pi v}{ \sqrt{x_-+\alpha_E}} E^{-1/4} < \frac{\pi v}{ \sqrt{\sqrt{2/3}\,v}} E^{-1/4} \; .
\]
It follows that $T(E)\to 0$ as $E \rightarrow +\infty$. 
\end{proof}

\subsection{Summary}
\noindent Summarizing the above steps, we have:
\begin{itemize}
\item
{\em Every} energy level $E>U(R_\ast)$ generates a unique closed orbit $\mathcal{O}_E$ in the $(R,R')$ phase plane, traced out by $R(\xi)$ in finite time. This orbit is strictly contained in $\{R>0\}$. 
\item
\emph{Every} non--equilibrium solution $R(\xi)$  with energy $E>U(R_*)$ is \emph{periodic}, oscillating between two finite turning points $R_{-}(E)<R_{+}(E)$ in a finite period $T(E)$ given by \eqref{eq:FullPeriod}.
\item
For each fixed energy $E>U(R_*)$, the orbit remains in the compact interval $[R_-(E),R_+(E)] \subset (0,\infty)$, but 
this confinement is not uniform: $R_-(E)\to 0^+$ and $R_+(E)\to +\infty$ as $E\to +\infty$.  
\item
As $E\downarrow U(R_*)$, these periodic orbits collapse onto the \emph{equilibrium} at $(R_*,0)$, and the period converges to $2\pi / \sqrt{U''(R_*)}$: this equilibrium is a \emph{center} with stable small--amplitude oscillations.
\item 
The phase portrait consists of: {\em (a)}  A \emph{unique} center equilibrium $(R_*,0)$ with $E=U(R_*)$. {\em (b)}  A nested family of \emph{closed} periodic orbits surrounding it for $E>U(R_*)$. 
{\em (c)} There are  no homoclinic/heteroclinic loops to other equilibria.
\end{itemize}

\section{Elliptic Integral Representation of Closed Orbits and Return to $(x,t)$ Variables}

Section 5 established the qualitative phase-plane picture for a subclass of travelling waves of the generalized Konno--Oono system, including the existence of periodic orbits for $E>U(R_*)$. 
We now describe the travelling--wave 
solutions corresponding to the periodic closed trajectories in the $(R,R')$ phase plane, and we express them through elliptic quadratures and Jacobi functions. We then revert to the original 
space--time variables $(x,t)$ and discuss their interpretation.

Let us fix an energy level $E>U(R_*)$ and, as in the previous section, we let $0 \;<\; R_{-}(E)  \;<\; R_{+}(E) \;<\; \infty$ 
be the two \emph{turning points} satisfying $U\bigl(R_{\pm}(E)\bigr) = E$, so that $R'(\xi)=0$ there. We recall that:
\begin{equation}
\label{eq:OrbitEquation}
R'(\xi)
\;=\;
\pm\sqrt{2\,\Bigl(E - U\bigl(R(\xi)\bigr)\Bigr)},
\end{equation}
the ``plus'' sign applying when $R(\xi)$ is increasing in $\xi$, and the ``minus'' sign when $R(\xi)$ is decreasing. 
Separate variables in \eqref{eq:OrbitEquation}:
\[
\frac{dR}{\,\sqrt{2\,\bigl(E - U(R)\bigr)}\,}
\;=\;
d\xi.
\]
In order to find an explicit formula, we integrate from one turning point up to a general $R(\xi)$:
\[
\xi - \xi_0
\;=\;
\int_{R_{-}}^{\,R(\xi)}
\!\!
\frac{dR}{\,\sqrt{\,2\bigl(E - \frac{R^4}{2\,v^2} \;-\; \frac{2\,C^2}{R^2}\bigr)}\,},
\]
where $\xi_0$ is an integration constant that we choose so that $R(\xi_0) = R_-$. We invert this integral as follows.
The energy equation can be written as 
\[
R'(\xi)^2 = 2E-\frac{R(\xi)^4}{v^2} -\frac{4C^2}{R(\xi)^2}
\]
and so we find, in terms of $y(\xi)=R(\xi)^2$, 
\begin{align}
y'(\xi)^2 = -\frac{4}{v^2}y(\xi)^3 +8E\,y(\xi) -16C^2\; .
\label{eq:y-Weierstrass}
\end{align}
If  $z(\xi)= - y(\xi)/v^2$,  equation \eqref{eq:y-Weierstrass} becomes
\begin{equation}
z'(\xi)^2 = 4z(\xi)^3-g_2z(\xi)-g_3 
\label{eq:Weierstrass-equation}
\end{equation}
with $g_2={8E}/v^2$ and $g_3={16C^2}/v^4$.  For $E>U(R_*)$, the discriminant of the cubic in (\ref{eq:Weierstrass-equation}) is
\begin{align*}
g_2^3-27g_3^2
&=
\frac{512E^3}{v^6}
-\frac{6912C^4}{v^8} =
\frac{512}{v^8}
\left(
E^3v^2-\frac{27}{2}C^4
\right)
>0,
\end{align*}
since $U(R_*)^3=27C^4/(2v^2)$.  Thus, it has three distinct real roots. We write
\[
4z^3-g_2z-g_3 = 4(z-e_1)(z-e_2)(z-e_3), \qquad e_1>e_2>e_3.
\]
Viete's relations imply that  $e_1>0>e_2>e_3$.  (Indeed, these relations yield
$e_1+e_2+e_3=0$, $e_1e_2+e_1e_3+e_2e_3=-g_2/4$, and $e_1e_2e_3=g_3/4$. Since $g_3>0$, the product $e_1e_2e_3$ is positive. But 
$e_1+e_2+e_3=0$, and so one root must be positive and the other two must be negative). 
We compute these roots: from subsection 5.2 we have that the turning points $R_-(E)$ and $R_+(E)$ are defined by $U(R_\pm)=E$,  $0<R_-<R_+$, 
and therefore $R_-^2$ and $R_+^2$ are the two positive roots of the cubic equation 
\[
y^3-2Ev^2y+4C^2v^2=0.
\]
It follows that  $e_2=- {R_-^2}/{v^2}$ and  $e_3=-{R_+^2}/{v^2}$.
If we set
\[
k^2=\frac{e_2-e_3}{e_1-e_3}\in(0,1),
\]
then the bounded real solution of Equation \eqref{eq:Weierstrass-equation} is
\[
z(\xi)
=
e_3+(e_2-e_3)
\operatorname{sn}^2
\left(
\sqrt{e_1-e_3}\,(\xi-\xi_0)+\mathbf K(k);k
\right),
\]
where \(\mathbf K(k)\) denotes the complete elliptic integral of the first kind. Thus, 
\begin{equation}   \label{eq:R-Jacobi}
R(\xi)^2
= - v^2 z(\xi) =
R_+^2 - \bigl(R_+^2-R_-^2\bigr) \operatorname{sn}^2 \left( \sqrt{e_1-e_3}\,(\xi-\xi_0)+\mathbf K(k);k \right)\; .
\end{equation}

\smallskip

Now we recall from the beginning of Section 5 that the third variable $\Delta(\xi)$ satisfies the equation
\begin{equation}
\label{eq:DeltaIntegration}
\Delta(\xi)
\;=\;
\Delta_0
\;+\;
C\,
\int_0^\xi
\frac{d\eta}{\,R(\eta)^2\,}.
\end{equation}
Because $R(\eta)$ is a bounded, strictly positive function, the integrand $\frac{1}{R(\eta)^2}$ is continuous and nonzero.  
We can convert \eqref{eq:DeltaIntegration} into a further elliptic integral by changing variables from 
$\eta$ to $R$, on each monotone branch, using \eqref{eq:OrbitEquation}.  Specifically, on each monotone branch, up to the sign, 
\[
\int 
\frac{d\eta}{R(\eta)^2}
\;=\;
\int 
\frac{dR}{R^2(\eta)}\,\frac{d\eta}{dR}
\;=\;
\int
\frac{dR}{R^2\,R'(\eta)}
\;=\;
\int
\frac{dR}{\,R^2\,
\sqrt{\,2\,\bigl(E - U(R)\bigr)}\,} \; .
\]
This is an elliptic integral, but it does not determine a periodic function: the function $\Delta(\xi)$ is not periodic when $C\neq 0$ because $\displaystyle \int_0^{T} \frac{d\eta}{R(\eta)^2} > 0$ 
over one oscillation period $T$. In fact, 
\[
\Delta(\xi+T)-\Delta(\xi)
=
C\int_\xi^{\xi+T}\frac{d\eta}{R(\eta)^2}
\neq 0
\qquad \text{if } C\neq 0.
\]
Hence $\Delta$ exhibits a nontrivial drift over each period of $R$.

\smallskip

Having determined $(R(\xi),Q(\xi),\Delta(\xi))$ in travelling--wave coordinates $\xi = x - v\,t$, we may write the full solutions of the PDE system explicitly:
\[
r(x,t)
\;=\;
R(x - v\,t),
\quad
q(x,t)
\;=\;
\frac{R\bigl(x - v\,t\bigr)^{2}}{v},
\quad
\delta(x,t)
\;=\;
\Delta\bigl(x - v\,t\bigr).
\]
Here $R(\cdot)$ is a smooth \emph{periodic} function of its single argument (with period $T(E)$ in $\xi$), while $\Delta(\cdot)$ is given by Equation (\ref{eq:DeltaIntegration}). 
Such  solutions fill out a continuous family parameterized by the energy $E>U(R_*)$ and the integration constants $(C \neq 0,\Delta_0,\xi_0)$. We have:
\begin{itemize}
\item 
\textit{Oscillatory travelling wave in $r(x,t)$.}  
Because $R(\xi+T)=R(\xi)$ for all $\xi$, the wave variable $r(x,t)= R(x- v t)$ is ``double-periodic'' in $(x,t)$ in the following sense: as $\xi$ varies by $T$, $r$ completes one oscillation.  
Equivalently, at any fixed time $t$, $r(x,t)$ is 
spatially periodic in $x$ of period $T$ (viewed in the travelling frame) and for each fixed $x$, it is periodic in $t$ with period $T/v$.
\item 
\textit{Corresponding $q(x,t)$.}  
We obtain $q(x,t)=R^2(x-vt)/v$, and so $q$ shares with $R$ the same period $T$ in $\xi$.  The function $q(\xi)$ is strictly positive since $R(\xi)>0$.
\item 
\textit{Phase drift of $ \delta(x,t)$.}  
The function $\Delta(\xi)$ produces a continuous drift in the ``phase'' variable $\delta(x,t)=\Delta(x-vt)$ as each cycle completes. This field is neither periodic nor bounded when $C\neq 0$.
This behavior is consistent with the conservative character of the system and the absence of any relaxation or damping mechanism.
\end{itemize}

We finish this section with the following geometric remarks:

\begin{itemize}
\item
\textbf{On intrinsic and mean curvature.}
For the subclass considered here, $C_0=0$ and hence \eqref{twgc} reduces to
\begin{equation}
K(\xi) = 1-\frac{2v^2C^2}{R(\xi)^6}\; . \label{Rstar0}
\end{equation}
On the other hand, the unique equilibrium point of the reduced potential satisfies $R_*^6=2C^2v^2$, and therefore,
\begin{equation}
K(\xi) = 1-\frac{R_*^6}{R(\xi)^6}\; .
\label{eq:curvature-Rstar}
\end{equation}

Let $E>U(R_*)$, and let $R_-(E)$ and $R_+(E)$ be the two turning points of the corresponding non-equilibrium periodic orbit.
Since $R_*$ is the unique minimum point of $U$, we have that $R_-(E)<R_*<R_+(E)$, see subsection 5.1.  It follows from \eqref{eq:curvature-Rstar} that
\[
K\bigl(R_-(E)\bigr)<0\; ,
\qquad
K(R_*)=0\; ,
\qquad
K\bigl(R_+(E)\bigr)>0\; .
\]
Thus, every non-equilibrium periodic orbit in the subclass $C_0=0$, $C\neq0$, determines an immersed surface whose Gaussian curvature changes sign periodically. More precisely, since
$R(\xi)$ crosses the value $R_*$ once on each monotone branch of the periodic oscillation, $K(\xi)$ vanishes exactly twice during each period of $R(\xi)$.

The mean curvature has an analogous periodic dependence on $\xi$. Indeed, for the subclass $C_0=0$, $C\neq0$, we have $Q(\xi)=R(\xi)^2/v$ and $\Delta'(\xi)=C/R(\xi)^2$, as
discussed at the beginning of Section 5, and therefore (\ref{curvatures}) yields
\begin{align*}
H(\xi)
&=
\frac{1}{R(\xi)Q(\xi)}
\left(
R(\xi)^2\Delta'(\xi)
+\frac{1}{4\lambda^2}\Delta'(\xi)
-\frac{v}{2}Q(\xi)\Delta'(\xi)
\right)
\\
&=
\frac{Cv}{R(\xi)^3}
\left(
\frac12+\frac{1}{4\lambda^2R(\xi)^2}
\right)\; . 
\end{align*}
Since $R(\xi)$ is periodic, so is $H(\xi)$. Moreover, for fixed $\lambda\neq0$ and $C\neq0$, the mean curvature does not vanish
and has the same sign as $Cv$. 
\item
\textbf{The small-amplitude limit.}
As \(E\downarrow U(R_*)\), the two turning points coalesce and the periodic travelling waves converge to the constant-amplitude solution
$R(\xi)\equiv R_*$, $R_*^6=2C^2v^2$. Equation (\ref{curvatures}) and the foregoing analysis imply that the Gaussian curvature at the equilibrium $R(\xi)\equiv R_*$ is
\[
K =
1-\frac{2C^2v^2}{R(\xi)^6} = 1-\frac{2C^2v^2}{R_*^6} = 0\; .
\]
On the other hand, the mean curvature becomes 
\[
H = \frac{Cv}{R_*^3} \left( \frac12+\frac{1}{4\lambda^2R_*^2} \right) \neq0\; .
\]
Thus, the limiting immersed surface is intrinsically flat and has nonzero constant mean curvature.

Interestingly, we note that if $k_1$ and $k_2$ denote the principal curvatures, then the Gauss equation in $S^3$ gives $K=1+k_1k_2$, and so $k_1k_2=-1$.
Moreover, since $H_*=(k_1+k_2)/2$, we also have $k_1+k_2=2H_*$. 
Hence, both principal curvatures are constant and distinct, and so the limiting surface is isoparametric in $S^3$. Thus, it must be locally congruent to
(an open subset of) a generalized Clifford torus $S^1(a)\times S^1(b)\subset S^3$, $a^2+b^2=1$.

\end{itemize}

\section{Numerical reconstruction of travelling waves}\label{sec:numerics}
We sketch a numerical scheme to compute the travelling--wave solutions of the generalized Konno--Oono system considered in Sections 5 and 6. 
We illustrate the oscillatory behavior of the function $R(\xi)$ and the monotonic behavior of $\Delta(\xi)$, as well as the corresponding reconstructed solutions 
$\bigl(r(x,t),\,q(x,t),\,\delta(x,t)\bigr)$ in physical space–time $(x,t)$.  

\subsection{Governing Equations and Setup}
The travelling--wave ansatz reduces the system \eqref{eq:1}-\eqref{eq:3} to equations  \eqref{eq:1'}-\eqref{eq:3'}.  We consider the subclass of travelling waves determined by
\[
R''(\xi)
\;=\;
-\;\frac{dU}{dR}\bigl(R(\xi)\bigr)\; ,
\qquad
U(R)
\;=\;
\frac{R^4}{2v^2}
\;+\;
\frac{2\,C^2}{R^2}\; ,
\qquad
\Delta'(\xi) \;=\;\frac{C}{\,R(\xi)^2\,} \; ,
\]
in agreement with Sections 5 and 6. The integration constants $v>0$ (wave speed) and $C\neq0$ determine the amplitude and phase‐drift behavior of the travelling wave.

\subsection{Implementation}
We employ \texttt{Python} and \texttt{scipy.integrate.solve\_ivp} to integrate the ODE system numerically.  An event detection mechanism captures the turning points of $R(\xi)$, 
enabling the solution to trace out multiple full periods:
\begin{enumerate}
  \item We bisect to find the turning points $R_{-}$ and $R_{+}$ satisfying $U(R_{\pm})=E$, with $E$ the total conserved energy.
  \item We integrate forward from $R_-$ to $R_+$, stopping precisely upon reaching $R_+$ (a zero of $R(\xi)-R_+$).
  \item We integrate backward from $R_+$ back to $R_-$, collecting a complete oscillation of $R(\xi)$.
  \item We repeat for a number of oscillations ($\texttt{num\_periods}$).
\end{enumerate}
Additionally, we map the travelling--wave solution back to $(x,t)$ by discrete sampling of $\xi = x - vt$, thus obtaining $\bigl(r(x,t),\,q(x,t),\,\delta(x,t)\bigr)$ on a chosen spatial and temporal domain.

\subsection{Numerical Profiles}

\begin{figure}[htbp]
\centering
\includegraphics[width=0.87\textwidth]{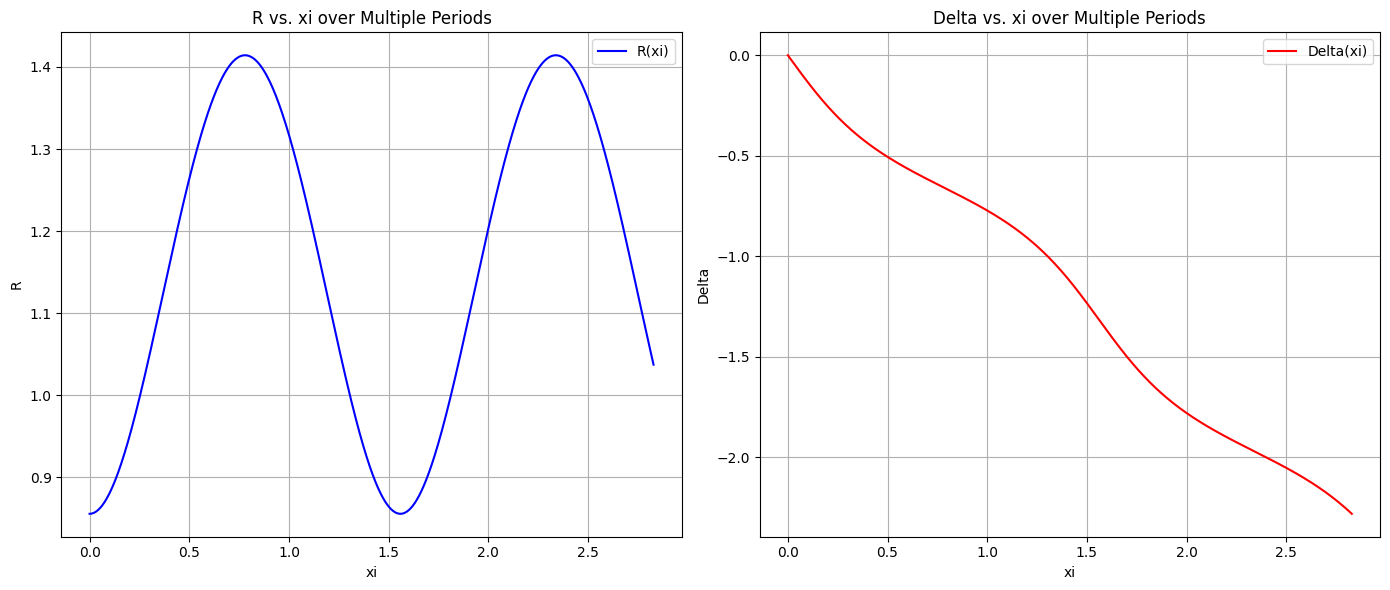}
\caption{(Left) The amplitude function $R(\xi)$ over one or more oscillation periods. (Right) The phase function $\Delta(\xi)$, 
which increases or decreases monotonically depending on the sign of $C$. In this example, $C<0$ leads to a strictly decreasing $\Delta(\xi)$.}
\label{fig:R_and_Delta}
\end{figure}

Figure~\ref{fig:R_and_Delta} shows the output of a single forward--backward integration (one full oscillation); we use $C=-1$ and $v=1$.  
The function $R(\xi)$ oscillates periodically between its turning points, while $\Delta(\xi)$ 
exhibits a drift. This figure also allows us to check the periodic sign change of the Gaussian curvature $K = 1 - 2/R^6$ given by (\ref{Rstar0}), from $K= \, \sim \!0.76$ to $K= \, \sim \!-4.3$. 


\begin{figure}[htbp]
\centering
\includegraphics[width=0.37\textwidth]{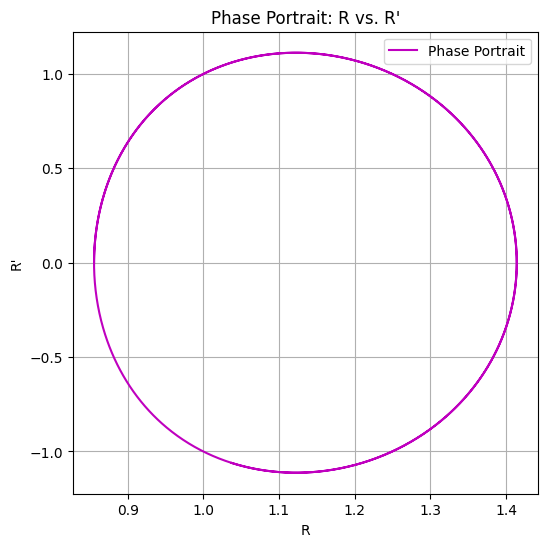}
\hfill
\includegraphics[width=0.5\textwidth]{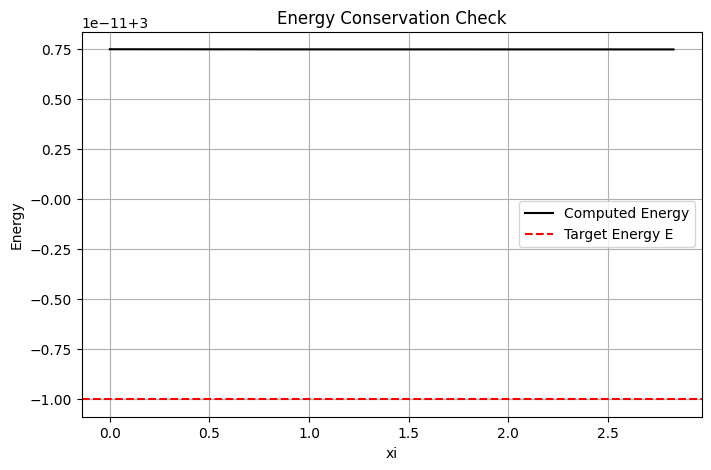}
\caption{(Left) Phase portrait $(R,R')$ illustrating a closed orbit characteristic of a conservative, single--well potential. (Right) Energy conservation check: the total energy $\tfrac12\,R'^2 + U(R)$ remains constant (horizontal dotted line) within numerical tolerances.}
\label{fig:PhasePortrait_Energy}
\end{figure}

In Figure~\ref{fig:PhasePortrait_Energy}, we illustrate two key characteristics of our system:  A closed loop in the $(R,R')$ plane confirming the periodic nature of $R(\xi)$ {\em (Phase Portrait)}; 
The numerically computed energy $\frac12\,R'^2 + U(R)$ remains close to the specified constant $E$, showing absence of numerical drift over multiple periods {\em (Energy Conservation)}. 

\begin{figure}[htbp]
\centering
\includegraphics[width=0.47\textwidth]{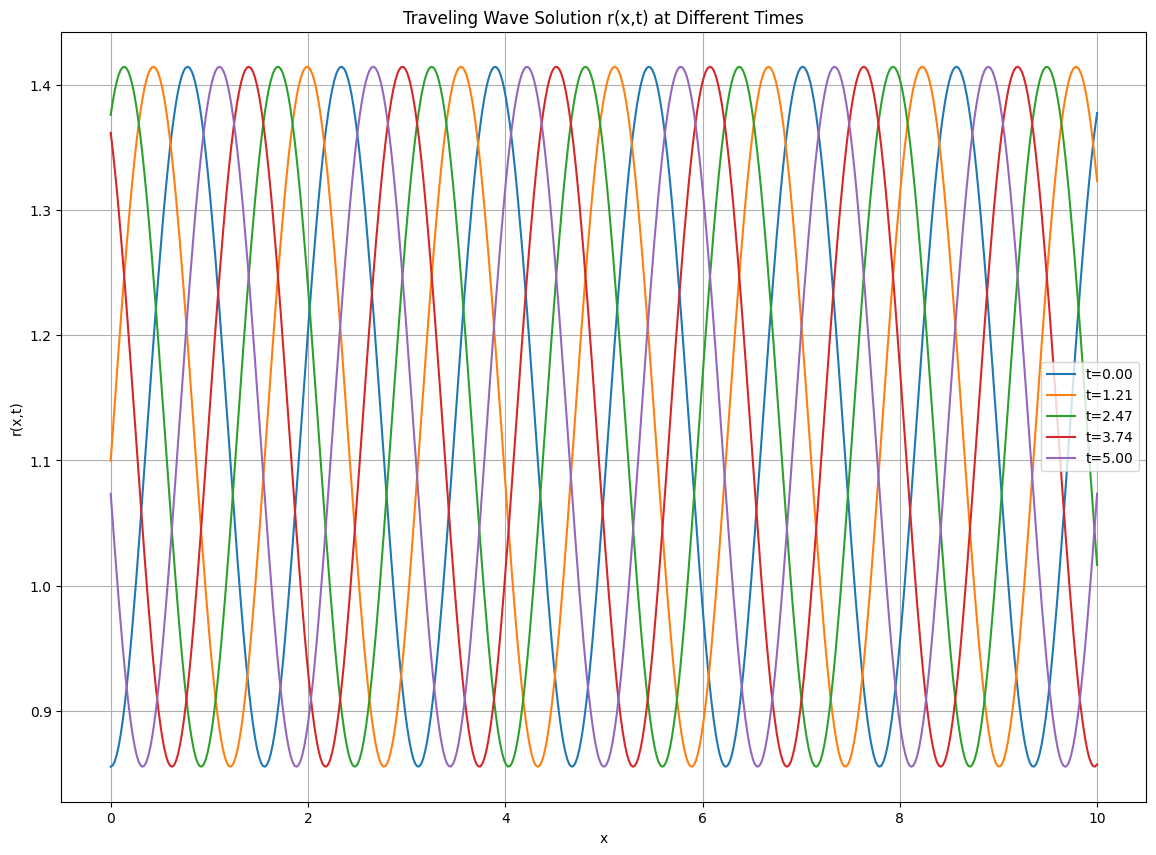}
\hfill
\includegraphics[width=0.47\textwidth]{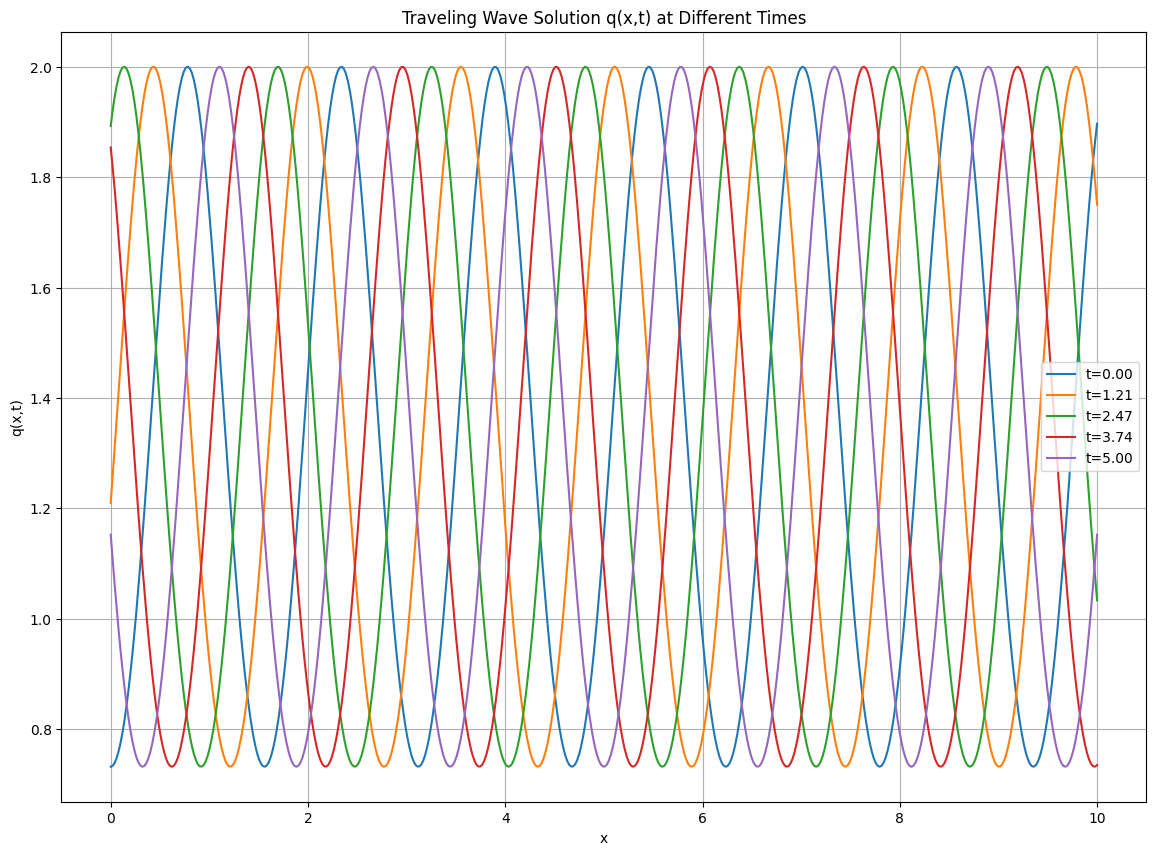}
\caption{Reconstructed travelling waves in original variables. (Left) The field $r(x,t)$. (Right) The field $q(x,t)=R^2(\xi)/v$. Both fields show spatially oscillatory patterns translated by speed $v$.}
\label{fig:r_q_xt}
\end{figure}

Finally, in Figure~\ref{fig:r_q_xt} we present a \emph{reconstruction} of the travelling--wave solution in $(x,t)$:
\[
r(x,t)=R(x - vt), 
\quad
q(x,t)=\frac{R(x - vt)^2}{v}.
\]
At each $(x,t)$, we evaluate $\xi\coloneqq x - vt$ and use $\bigl(R(\xi),\,\Delta(\xi)\bigr)$ as obtained from the ODE integration.  We observe a \emph{periodic wave} in space that advances rigidly with velocity $v$.  

\subsection{Phase Drift Observation}
As explained in Section 6, the quantity $\Delta(\xi)$ is strictly monotonic and accumulates a net increment $\Delta(\xi+T) - \Delta(\xi) \;\neq\;0$, 
producing a constant phase drift. We present two illustrations.

\begin{figure}[htbp]
\centering
\includegraphics[width=0.47\textwidth]{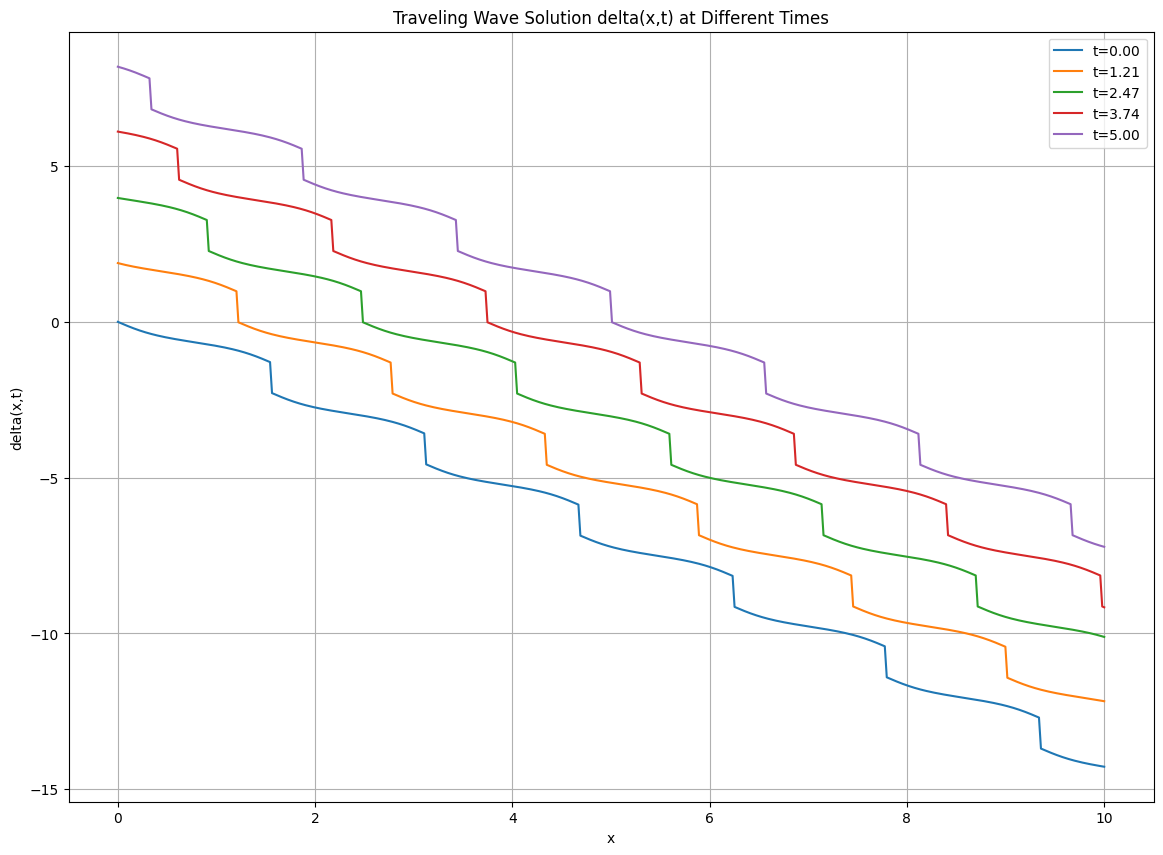}
\caption{Space--time reconstruction of $\delta(x,t)$ with $C<0$.  Since $\Delta'(\xi)=C/R^2(\xi)$, the field $\delta$ shows a monotonic decay.} 
\label{fig:delta_xt}
\end{figure}

\begin{figure}[htbp]
\centering
\includegraphics[width=0.47\textwidth]{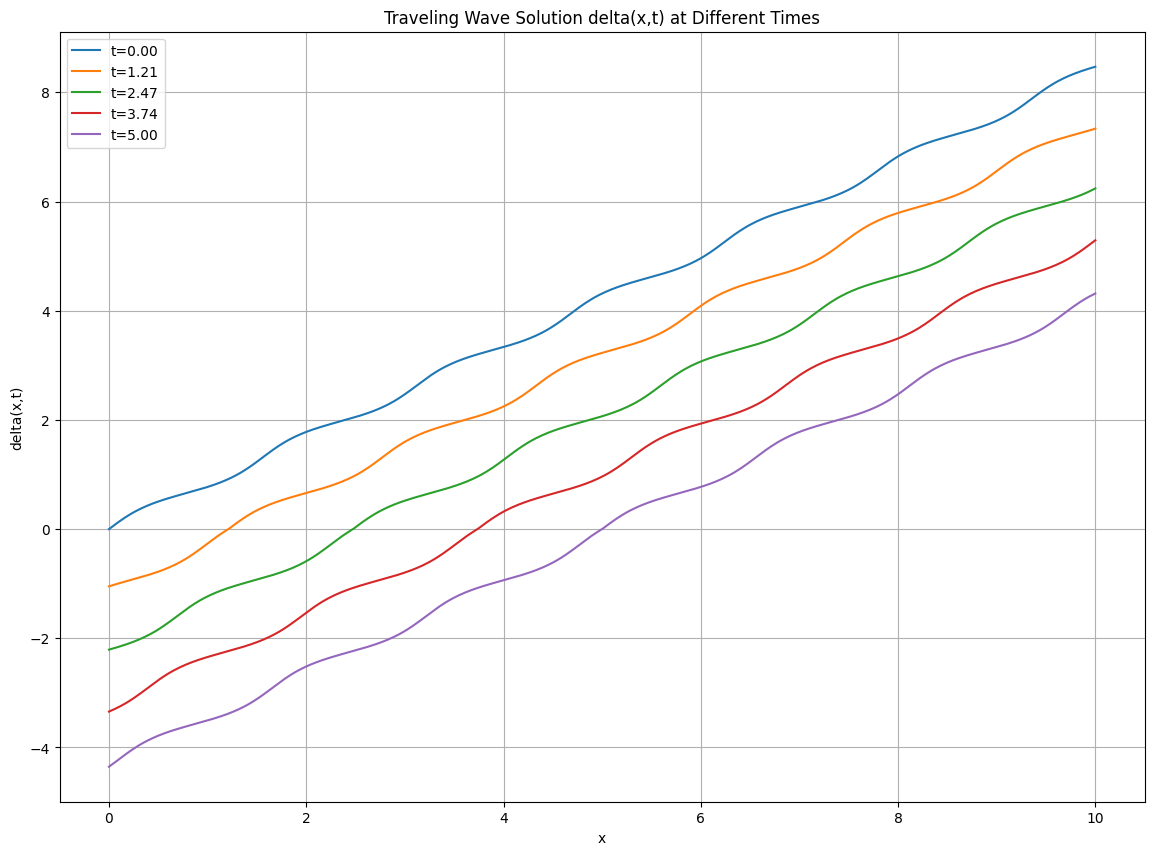}
\caption{Space--time reconstruction of $\delta(x,t)$ with $C>0$.  Since $\Delta'(\xi)=C/R^2(\xi)$, the field $\delta$ shows a monotonic growth.}
\label{fig:delta_xt1}
\end{figure}



\newpage

\paragraph{Acknowledgments:}
E.G.R.'s work is partially supported by the research project FONDECYT \# 1241719.

\end{document}